\documentclass[11pt]{article}
\usepackage{amssymb,amsmath,amsthm,epsfig}
\usepackage{txfonts}
\usepackage{latexsym, enumerate,cite}
\usepackage{eepic}
\usepackage{booktabs}
\usepackage{afterpage}
\usepackage{epic}
\usepackage{graphicx}
\usepackage{color}
\usepackage{ifpdf}
\usepackage{subfigure}
\usepackage{tikz}
\usepackage{dsfont}
\usepackage{multirow}
\usepackage{makecell}
\usepackage{algorithm}
\usepackage{bm}
\usepackage{multirow}
\usepackage{color}
\usepackage[colorlinks, linkcolor=blue,anchorcolor=blue,citecolor=blue,urlcolor=blue]{hyperref}
\usepackage{appendix}
\usepackage{comment}
\graphicspath{{figs/}}

\numberwithin{equation}{section}

\theoremstyle{plain}
\newtheorem{lem}{Lemma}[section]
\newtheorem{thm}[lem]{Theorem}

\theoremstyle{definition}
\newtheorem{defn}{Definition}[section]

\newtheorem{rem}{Remark}[section]

\DeclareMathOperator{\diag}{diag}

\newcommand{\ddp}[2]{\frac{\partial#1}{\partial#2}}

\newcommand{\p}{\partial}
\newcommand{\ds}{\displaystyle}

\newcommand{\RR}{\mathbb{R}}
\newcommand{\rmr}{\mathrm{r}}

\newcommand{\nm}{\noalign{\smallskip}}

\newcommand{\Scal}{\mathcal{S}}

\newcommand{\Kcal}{\mathcal{K}}

\renewcommand{\(}{\left(}
\renewcommand{\)}{\right)}
\newcommand{\norm}[1]{\left\Vert#1\right\Vert}

\def \i{\ensuremath{\mathrm{i}}}
\def \d{\ensuremath{\mathrm{d}}}
\begin{document}

\title{Mathematical theory on multi-layer high contrast acoustic subwavelength resonators}
\author{
{Youjun Deng}
\thanks{School of Mathematics and Statistics, Central South University, Changsha, 410083, Hunan Province, China. \ \ Email: youjundeng@csu.edu.cn; dengyijun\_001@163.com}
\and
{Lingzheng Kong}
\thanks{Corresponding author. School of Mathematics and Statistics, Central South University, Changsha, 410083, Hunan Province, China. Email: math\_klz@csu.edu.cn; math\_klz@163.com}
\and
{Hongjie Li}
\thanks{Yau Mathematical Sciences Center, Tsinghua University, Beijing, China. \ \  Email: hongjieli@tsinghua.edu.cn; hongjie\_li@yeah.net.}
\and
{Hongyu Liu}
\thanks{Department of Mathematics, City University of Hong Kong, Hong Kong SAR, China. \ \ Email: hongyliu@cityu.edu.hk}	
\and
{Liyan Zhu}
\thanks{School of Mathematics and Statistics, Central South University, Changsha, 410083, Hunan Province, China.\ \ Email: math\_zly@csu.edu.cn}
}
\date{}%
\maketitle
\begin{abstract}
Subwavelength resonance is a vital acoustic phenomenon in contrasting media. The narrow bandgap width of single-layer resonator has prompted the exploration of multi-layer metamaterials as an effective alternative,
which consist of alternating nests of high-contrast materials, called ``resonators'', and a background media.
In this paper, we develop a general mathematical framework for studying acoustics within multi-layer high-contrast structures. Firstly, by using layer potential techniques, we establish the representation formula in terms of a matrix type operator with a block tridiagonal form  for  multi-layer structures within general geometry. Then we prove the existence of subwavelength resonances via Gohberg-Sigal theory, which generalizes the celebrated Minnaert resonances in single-layer structures. Intriguingly, we find that the primary contribution to mode splitting lies in the fact that as the number of nested resonators increases, the degree of the corresponding characteristic polynomial also increases, while the type of resonance (consists solely of monopolar resonances) remains unchanged. Furthermore, we derive original formulas for the subwavelength resonance frequencies of concentric dual-resonator.  Numerical results associated with different nested resonators are presented to corroborate the theoretical findings. 
\end{abstract}



\section{Introduction}
In recent decades, plasmonic resonant structures have been extensively studied and utilized as building blocks to make novel optical and acoustic devices \cite{DLZJMPA21,PRHN2003SCI,DLbook2024,YA_PNAS19,YA_SIAMREV18,LLPRSA18, ALLZ526}. These plasmonic materials are typically composed of noble metals, which may exhibit negative properties under specific conditions. It has been mathematically demonstrated that the plasmonic resonance can be formulated  as an eigenvalue problem of the Neumann-Poincar\'e operator \cite{Ammari2013,KKLSY_JLMS2016,JKMA23,BZ_RMI2019,Ammari2016,BLLW_ESAIM2020,ACLJFA23, DLL242, L1272}. However, metallic structures inherently  suffer from high losses and dissipation due to heating, which severely limits the efficiency and functionality of plasmonic devices. This limitation motivates the exploration of alternatives to metallic subwavelength resonators.

Recent developments in microscale acoustic physics have led to a new branch of  phononic crystal focused on the manipulation of acoustically induced subwavelength resonances in high-contrast resonators.
Resonant high-contrast microstructures form new building blocks, which can be used to realize negative materials through homogenization theory in specific configurations \cite{AmmariSIMA17, AFLYZQAM19}. In particular, when bubbles are embedded in liquids, even a small volume fraction of bubbles can significantly impact the velocity of waves in liquids \cite{CAJASA1989}. This phenomenon is mainly due to the high density contrast between the bubbles and the background liquid, causing the bubbles to oscillate vigorously. At a specific low frequency called the Minnaert resonant frequency, the bubbles can act as acoustic resonators \cite{Min_1933}. While bubbly media composed of air bubbles in water exhibit intriguing properties for the creation of subwavelength metamaterials \cite{AD_book2024,AFLYZ_JDE2017,AmmariSIMA17}, such structures tend to be highly unstable \cite{LVAPL09}. Various strategies have been proposed to stabilize these structures. One approach involves encapsulating the bubbles in a thin coating (e.g., albumin, polymer, or lipid) with the aim of mitigating the rapid dissolution and coalescence of the bubbles \cite{DB_IEEE2011}. Another approach is to substitute the background medium, water, with a soft elastic material. It has been demonstrated that this technique yields metamaterials with properties analogous to those of air bubbles in water \cite{LSPS_JASM2008}. Recently, several mathematical theories have been developed to enhance the understanding of the Minnaert resonance of bubbles. Due to the high density contrast between air bubbles and the liquid or elastic medium, the authors in \cite{AFGLZ_AIHPCAN, AFHLY_JDE2019, LLZSIAM2022} have conducted rigorous and systematic mathematical studies of Minnaert resonance for single bubbles encapsulated in thin coatings, immersed in liquids, and in soft elastic materials, respectively. Furthermore, when hard inclusions, such as rubber-coated epoxy spheres \cite{ZH_PRB2009} or steel-coated soft silicone rubber \cite{LWSZ_NM2011}, are embedded in soft elastic materials, dipolar resonance is induced within the subwavelength regime. This phenomenon has been experimentally realized \cite{LZScience} and mathematically derived for the first time in \cite{LZArxiv, LL7527}. As mention above, one significant application of subwavelength resonance and contrasting material structures is the effective realization of various metamaterials with negative material properties. Based on this realization, a class of phononic crystals, made from  periodic arrangements of separated subwavelength resonators\cite{AFHLY_SIMA2020, AFLYZ_JDE2017, AK_book2018}, that exhibits bandgaps and has been employed in advanced techniques for manipulating wave propagation at subwavelength scales.

In most studies on phononic crystals, the structures are typically composed of single-layer (homogeneous) resonators. Due to their narrow bandgap width and poor wave filtering performance, such configurations are not readily applicable in practical engineering contexts \cite{PVDD_SCR2010}. This limitation has prompted the design and investigation of metamaterials with wider bandgaps. In particular, multi-layer high-contrast metamaterials have emerged as a popular choice for subwavelength resonators, owing to their high tunability and high quality resonance. Experimental and numerical observations \cite{LPDV_PRE2007,CMJW_SV2018,KMKG_JVA2017,SPMW_AA2023} indicate that multi-layer concentric radial resonators, at the subwavelength regime, can open multiple local resonance bandgaps. Meanwhile, by appropriately combining multiple multi-layer concentric radial structures, it is possible to overlap sharp dips (strong field concentration) and create a larger acoustic stop band. However, despite considerable evidence in the engineering and physics literature, the mathematical understanding of the origin of subwavelength resonance in multi-layer contrasting media and the mechanism underlying mode splitting (the separation of subwavelength resonant frequencies) remains limited, with no quantitative results available even for concentric dual-resonator. This prompts us to demonstrate the opening of multiple bandgaps in multi-layer subwavelength resonators. We consider multi-layer metamaterials characterized by a nested structure similar to Russian folk art dolls called Matryoshka dolls and exploit their subwavelength resonance. The number of layers can be arbitrary and the material parameters in each layer may be different, though uniform. High density contrast is crucial for achieving resonance at subwavelength scales. The wave propagation in the multi-layer structure is modeled by a high-contrast Helmholtz problem. It is noteworthy that, owing to advancements in 3D printing techniques \cite{MBKPD_PNAS2016, CPBE2018}, the compositional structure of high-contrast materials has become increasingly diverse (e.g., resin materials \cite{DSWYZ_APL2021}, polymer materials \cite{SPMW_AA2023}, etc.), extending beyond the previously mentioned Minnaert cavities and associated stabilization strategies. Such multi-layer high contrast materials arise naturally when designing subwavelength metamaterials, however, due to the complex structure of multi-layer acoustic metamaterial systems, fully determining and rigorously demonstrating the mechanisms behind their acoustic properties is a rather challenging task.

To establish the primary conclusion of this paper, we first make use of the layer potential techniques to reduce the acoustic scattering problem into a system of integral matrix type operator $\mathcal{A}(\omega,\delta)$ (see \eqref{contrastintegralsystem}--\eqref{contrastinteract}) having the block tridiagonal form in the $N$-layer structure with $C^{1,\eta}\, (0<\eta<1)$ smooth interface. Secondly, by using the asymptotic perturbation and Gohberg-Sigal theory\cite{AK_book2018},  we demonstrate that $\mathcal{A}(\omega(\delta),\delta)$ has $2N_\rmr$ characteristic values that are symmetric about the imaginary axis, with $\omega(\delta)$  depending  continuously on $\delta$ and $\omega(\delta)\to 0$ as the material contrast $\delta\to 0$. In fact, $N_\rmr:=\left\lfloor (N+1)/2\right\rfloor$ here represents the number of resonator elements in the $N$-layer structure. In other words, the number of bandgap increases with the number of resonator-nested. It is worth noting that the primary reason for mode splitting lies in the fact that as the number of nested resonators increases, the degree of the corresponding characteristic polynomial also increases, while the type of resonance (which consists solely of monopolar resonances) remains unchanged. It is known that the resonant frequency is associated with the shape of the resonators \cite{AFGLZ_AIHPCAN}. However, the rotational symmetry breaking of the resonators does not lead to mode splitting. For this, based on Fourier series, we present an exact matrix representation of $\mathcal{A}(\omega,\delta)$ in  multi-layer concentric balls. By highly intricate and delicate analysis, we derive exact and original formulas for the resonant frequency of  concentric balls with layers $2\le N\le 4$. For structures with a large number of layers, we shall provide  numerical computations of resonant modes.
In practical applications, the multi-layer high contrast  structure can serve as a fundamental building block for various material devices. Our analysis will provide a powerful and general design principle that can be applied to select appropriate material parameters, both qualitatively and quantitatively, guiding the design of resonator-nested  and predicting their resonant properties. By adjusting the geometries and parameters of the materials, one can fine-tune the resonant frequencies to target specific acoustic applications. We shall investigate along this direction in our forthcoming work.

The remainder of this paper is organized as follows. In Section \ref{sec2}, we first present some preliminary knowledge on boundary layer potentials and then establish the representation formula of the solution of the acoustic scattering problem with multi-layer structures. Section \ref{sec3} is devoted to the study of subwavelength resonance for multi-layer high contrast metamaterials by using the Gohberg-Sigal theory. In Section \ref{sec4}, the exact  formulas of the resonant frequencies  for single-resonator, dual-resonator models is given. In section \ref{sec5}, numerical computations are presented in finding all the resonance modes for fixed structures with a large number of layers. Moreover, the strong field concentration is observed. Some concluding remarks are made in Section \ref{sec6}.

\section{Preliminaries}\label{sec2}
\subsection{Layer potentials}\label{subsec21}

Our study of subwavelength resonance within the Helmholtz system relies heavily on layer potential theory. Thus, we briefly introduce the boundary layer potential operators and associated properties \cite{CK_book}.

Let $G_k$ denote the outgoing fundamental solution to the PDO $ \Delta+k^2$ in $\mathbb{R}^3$, defined as
\begin{equation}\label{fundamentalk}
G_k(x)= - \frac{e^{\i k|x|}}{4 \pi|x|}.
\end{equation}
Let $D$ be a bounded domain with a $C^{1,\eta}\, (0<\eta<1)$ boundary $\Gamma$. The single layer potential $\mathcal{S}_{\Gamma}^{k}$ associated with wavenumber $k$ is defined by
 \begin{equation}\label{eq:sh}
 \mathcal{S}_{{\Gamma}}^{k} [\phi](x) =  \int_{{\Gamma}} G_k(x- y) \phi(y) ~\mathrm{d}\sigma(y),  \quad x \in  \mathbb{R}^3,
 \end{equation}
where $\phi\in L^2(\Gamma)$ is the density function. There hold the following jump relations on the surface \cite{Ammari2007}
\begin{equation} \label{singlejumpk}
\frac{\p}{\p\nu}\Scal^k_{\Gamma} [\phi] \Big|_{\pm} = \(\pm \frac{1}{2}I+
\Kcal_{{\Gamma}}^{k,*}\)[\phi] \quad \mbox{on } {\Gamma},
\end{equation}
where the subscripts $+$ and $-$ denote evaluation from outside and inside the boundary $\Gamma$, respectively. In \eqref{singlejumpk}, the operator $\Kcal_{\Gamma}^{k,*}$ is called the Neumann-Poincar\'e (NP) operator defined by
\[
\Kcal_{\Gamma}^{k,*}[\phi]
(x) = \mbox{p.v.}\;\int_{{\Gamma}}\frac{\p G_k(x-y)}{\p \nu_x}\phi(y)~\mathrm{d} \sigma(y),\quad x \in  \Gamma,
\]
where p.v. stands for the Cauchy principle value. In what follows, we denote by $\Scal_{\Gamma}$ and $\Kcal_{\Gamma}^{*}$ be the single-layer and Neumann-Poincar\'e operators $\Scal_{\Gamma}^k$ and $\Kcal_{\Gamma}^{k,*}$, by formally taking $k =0$ respectively.

Since we are interested in low-frequency regime, we will use the following asymptotic expansion \cite{AFGLZ_AIHPCAN}:
\begin{equation} \label{series-s}
\mathcal{S}_{\Gamma}^{k}=  \mathcal{S}_{\Gamma} + \sum_{j=1}^{\infty} k^j \mathcal{S}_{\Gamma, j},
\end{equation}
where
\[
\mathcal{S}_{\Gamma, j} [\phi](x) = - \frac{\i}{4 \pi} \int_{\Gamma} \frac{ (\i |x-y|)^{j-1}}{j! } \phi(y)~\d\sigma(y).
\]
In particular, we have
\begin{equation}
\mathcal{S}_{\Gamma, 1} [\phi](x) = - \frac{\i}{4 \pi}  \int_{\Gamma}  \phi(y)~\d\sigma(y).
\end{equation}
It is well known that $\mathcal{S}_{\Gamma}:
L^2(\Gamma) \rightarrow H^1(\Gamma)$ is invertible \cite{Ammari2007}.

Similarly, the NP operator $\mathcal{K}_{\Gamma}^{k, *}$ has the following asymptotic expansion
\begin{equation} \label{series_K}
\mathcal{K}_{\Gamma}^{k,*}  = \mathcal{K}_\Gamma^* + \sum_{j=1}^{\infty} k^j \mathcal{K}_{\Gamma, j},
\end{equation}
where
$$
\mathcal{K}_{\Gamma, j}[\phi](x) = - \frac{\i}{4 \pi} \int_{\Gamma} \frac{ \p (\i|x-y|)^{j-1}}{j! \p \nu(x)} \phi(y) d\sigma(y)=
- \frac{\i^j (j-1)}{4 \pi j!} \int_{\Gamma} |x-y|^{j-3} (x-y)\cdot\nu(x) \phi(y) ~\d\sigma(y).
$$
In particular, we have
\begin{align}
\mathcal{K}_{\Gamma, 1} &=0,\\
\mathcal{K}_{\Gamma, 2}[\phi](x) &= \frac{1}{8\pi} \int_{\Gamma} \frac{(x-y)\cdot \nu(x)}{|x-y|} \phi(y)~\d\sigma(y),
\end{align}
\begin{equation}\label{rank2}
\mathcal{K}_{\Gamma, 2}^*[1] (x)= \frac{1}{8\pi}\int_{\Gamma} \frac{(y-x)\cdot \nu(y)}{|y-x|}~\d\sigma(y)
= \frac{1}{8\pi}\int_{D} \nabla \cdot \frac{y-x}{|y-x|}~\d y =  \frac{1}{4\pi}\int_{D} \frac{1}{|y-x|}~\d y.
\end{equation}

\begin{lem}[see \cite{AK_book2018}] \label{lem-appendix11}
	The norms
	$\| \mathcal{S}_{\Gamma, j} \|_{
		\mathcal{L}(L^2(\Gamma), H^1(\Gamma))}$
	and $\| \mathcal{K}_{\Gamma, j} \|_{ \mathcal{L}(L^2(\Gamma))}$
	 are uniformly bounded with respect to $j$. Moreover, the series in (\ref{series-s}) and \eqref{series_K} are convergent in
	$\mathcal{B}(L^2(\Gamma), H^1(\Gamma))$ and $\mathcal{B}(L^2(\Gamma),L^2(\Gamma))$, respectively.
\end{lem}

\subsection{Acoustics with multi-layer structures}
In this subsection, we first introduce the material configurations for the subsequent analysis. Before studying of high contrast material, let us first focus on regular material of multi-layer structure, defined as follows.
\begin{defn}\label{def21}
	Let $D$ be a domain with a $C^{1,\eta}$ smooth boundary, denoted as $\Gamma_1 := \partial D$, and let $D_0 = \mathbb{R}^3 \setminus \overline{D}$. We say that $D$ has a partition with a multi-layer structure if its interior is divided by closed, non-intersecting, and well-separated $C^{1,\eta}$ surfaces $\Gamma_j$ $(j = 2, 3, \dots, N)$ into subsets (layers) $D_j$ $(j = 1, 2, \dots, N)$. Each surface $\Gamma_{j-1}$ surrounds the surface $\Gamma_j$ $(j = 2, 3, \dots, N)$. The region $D_j$ $(j = 0, 1, 2, \dots, N)$ represents homogeneous media.
\end{defn}

For the $N$-layer structure $D$ specified in Definition \ref{def21}, the corresponding physical parameters, i.e. the density and the bulk modulus, are given by
\begin{equation}\label{densityA&AC}
\rho(x)=\rho_c(x)\chi(D)+\rho_0\chi(\mathbb{R}^2\backslash \overline{D}),
\end{equation}
\begin{equation}\label{bulkA&AC}
\kappa(x)=\kappa_c(x)\chi(D)+\kappa_0\chi(\mathbb{R}^2\backslash \overline{D}),
\end{equation}
where $\chi$ stands for the characteristic function of a domain. In the last formula, the parameters $\rho_c$ and $\kappa_c$ enjoy the following form
\begin{equation}\label{conductvt_each_layer}
\rho_c(x)=\rho_j\;\mbox{ and } \kappa_c = \kappa_j, \quad x\in D_j,\quad j=1,2,\ldots,N.
\end{equation}
Moreover, it is assumed that
\begin{equation}\label{layercondition}
\rho_j\neq \rho_{j-1} \;\mbox{ and }\;\kappa_j\neq \kappa_{j-1},\;\mbox{ for } j=1,2,\ldots,N,
\end{equation}
which means that $D$ is of a layered-piecewise constant structure.
We then introduce the auxiliary parameters to facilitate our analysis:
\[
v_j=\sqrt{\frac{\kappa_j}{\rho_j}}\;\mbox{ and } \; k_j=\frac{\omega}{v_j},
\]
which are the wave speeds and wavenumbers in $D_j$, $j=0,1,\ldots,N$, respectively.

In summary, we consider a rather general multi-layer structure in which the number of layers can be arbitrarily given and the material parameters in each layer can be different from one another.

We will consider the scattering of a time-harmonic acoustic wave by the multi-layer structure, described by the following system
\begin{equation} \label{wave_equation}
\begin{cases}
 \ds\nabla\cdot\frac{1}{\rho_j}\nabla u  + \frac{\omega^2}{\kappa_j} u = 0, & \text{in  } D_j, \; k =0,1,\ldots,N, \\
 \nm
\ds u|_+ - u|_- = 0, & \text{on } \Gamma_j, \; j =1,2,\ldots,N,\\
\nm
\ds \frac{1}{\rho_{j-1}} \ddp{u}{\nu_j}|_+ - \frac{1}{\rho_j} \ddp{u}{\nu_j}|_- = 0, & \text{on }\Gamma_j,\; j =1,2,\ldots,N, \\
\nm
\ds u^s := u - u^{in} &  \mbox{satisfies the Sommerfeld radiation condition,}
\end{cases}
\end{equation}
where $u^{in}$ is the incoming wave, and  the notation $\nu_j$ denotes the outward normal on $\Gamma_j$.
By the Sommerfeld radiation condition, the scattered wave $u^s$ satisfies
\begin{equation} \label{eq:src}
\left(\ddp{}{|x|}-\i  k_0\right)u^s=O(|x|^{-2})\quad \mbox{ as }|x|\to\infty.
\end{equation}

With the help of the layer potentials in subsection \ref{subsec21}, the solution to the Helmholtz system \eqref{wave_equation} can be written by
\begin{equation}\label{Helm-solution}
u(x) = \begin{cases}
\ds u^{in} + \mathcal{S}_{\Gamma_1}^{k_0} [\psi_1](x), & \quad x \in D_0,\\
\nm
\ds\mathcal{S}_{{\Gamma_j}}^{k_j} [\phi_j](x)+\mathcal{S}_{{\Gamma_{j+1}}}^{k_j} [\psi_{j+1}](x),  & \quad x \in {D_j},\; j = 1,2,\ldots,N-1,\\
\nm
\ds \mathcal{S}_{{\Gamma_N}}^{k_N} [\phi_N](x) ,  & \quad x \in {D_N},
\end{cases}
\end{equation}
where $\psi_j,\phi_j \in  L^2(\Gamma_j)$, $j =1,2,\ldots,N.$ Using the second and third conditions in \eqref{wave_equation}, and the jump relations for the single layer potentials, we can obtain that $\psi_j$ and $\phi_j$, $j =1,2,\ldots,N,$  satisfy the following system of boundary integral equations:
\begin{equation}\label{integralsystem}
\mathcal{A}[\Psi]=F,
\end{equation}
where
\[
\Psi= (\psi_1,\phi_1,\psi_2,\phi_2,\ldots,\psi_{N},\phi_N)^T,
\,\,F= (u^{in}, t_0 \frac{\partial u^{in}}{\partial \nu_1},0,0,\ldots,0)^T,\; \mbox{ and } t_{i-1} = \frac{\rho_{i}}{\rho_{i-1}}\;\mbox{ for }\;  i = 1,\ldots,N.
\]
The $2N$-by-$2N$ matrix type operator $\mathcal{A}$ has the block tridiagonal form
\begin{equation}\label{layeredintegralsystem}
\begin{split}
\mathcal{A}:=\diag\left(\mathcal{L}_{i,i-1},\mathcal{M}_i, \mathcal{L}_{i,i+1}\right):
&=\begin{pmatrix}
\mathcal{M}_1 & \mathcal{L}_{1,2} & &  & & \\
\mathcal{L}_{2,1} & \mathcal{M}_2 &\mathcal{L}_{2,3}&  & &\\
&\mathcal{L}_{3,2} & \mathcal{M}_3 & \mathcal{L}_{3,4} & &\\
& & \ddots &\ddots & \ddots& \\
&  & & \mathcal{L}_{N-1,N-2} & \mathcal{M}_{N-1} & \mathcal{L}_{N-1,N}\\
&  & &  &\mathcal{L}_{N,N-1} & \mathcal{M}_N
\end{pmatrix}.
\end{split}
\end{equation}
In the last formula, $\mathcal{M}_i$ is the self-interaction for the $i$-th interface  defined by
\begin{equation}\label{selfinteract}
\begin{split}
 \mathcal{M}_i:=\begin{pmatrix}
-\mathcal{S}^{k_{i-1}}_{\Gamma_i} & \mathcal{S}^{k_{i}}_{\Gamma_i} \\
-t_{i-1}( \frac{1}{2}I+ \mathcal{K}_{\Gamma_i}^{k_{i-1}, *}) & -\frac{1}{2}I+ \mathcal{K}_{\Gamma_i}^{k_i, *}
\end{pmatrix},
\end{split}
\end{equation}
and for $|i-j|=1$, $\mathcal{L}_{i,j}$ encodes the effect of the $j$-th interface  to the $i$-th interface as defined respectively by
\begin{equation}\label{aobveinteract}
\begin{split}
\quad \mathcal{L}_{i,i+1}:=\begin{pmatrix}
\ds \mathcal{S}^{k_{i}}_{\Gamma_{i,i+1}}& 0\\
\nm
\ds \mathcal{K}_{\Gamma_{i,i+1}}^{k_i, *}& 0
\end{pmatrix},
\end{split}
\end{equation}
and
\begin{equation}\label{belowinteract}
\begin{split}
\quad \mathcal{L}_{i,i-1}:=\begin{pmatrix}
\ds 0&-\mathcal{S}^{k_{i-1}}_{\Gamma_{i,i-1}}\\
\nm
\ds 0&-t_{i-1}\mathcal{K}_{\Gamma_{i,i-1}}^{k_{i-1}, *}
\end{pmatrix}.
\end{split}
\end{equation}
Here, we introduced the operators $\mathcal{S}^{k}_{\Gamma_{i,j}}:L^2(\Gamma_{j})\to L^2(\Gamma_i)$ and $\mathcal{K}_{\Gamma_{i,j}}^{k, *}:L^2(\Gamma_{j})\to L^2(\Gamma_i)$ defined by
\[
\mathcal{S}^{k}_{\Gamma_{i,j}}[\varphi] = \mathcal{S}^{k}_{\Gamma_{j}}[\varphi] \big|_{\Gamma_i}\;\mbox{ and }\; \mathcal{K}_{\Gamma_{i,j}}^{k, *}[\varphi] = \ddp{}{\nu_i}\mathcal{S}^{k}_{\Gamma_{j}}[\varphi] \big|_{\Gamma_i}, \mbox{ for } \forall \varphi\in L^2(\Gamma_{j}).
\]
For the further discussion, we introduce the spaces $\mathcal{H} =\prod_{j=1}^{N} (L^2(\Gamma_j) \times L^2(\Gamma_j))$ and by $\mathcal{H}_1 =\prod_{j=1}^{N} (H^1(\Gamma_j) \times L^2(\Gamma_j))$,
It is clear that $\mathcal{A}(\omega, \delta)$ is a bounded linear operator from $\mathcal{H}$ to $\mathcal{H}_1$, i.e.
$\mathcal{A}(\omega, \delta) \in \mathcal{B}(\mathcal{H}, \mathcal{H}_1)$.

The following result guarantees the unique solvability of integral system \eqref{integralsystem} and consequently the well-posedness of problems \eqref{wave_equation}.

\begin{thm}\label{well-posedness}
Let $D$ be the $N$-layer configuration defined in Definition \ref{def21} with parameters satisfying \eqref{densityA&AC}--\eqref{layercondition}.
Suppose that $k_j^2$ is not a Dirichlet eigenvalue for $-\Delta$ in $\cup_{i=j+1}^N D_{i}$, $j = 0,1,\ldots,N-1$.  For any function $u^{in}\in H^1(\Gamma_1)$, there exists a unique solution $\Psi\in \mathcal{H}$ to the integral system \eqref{integralsystem}. Moreover, there exists a constant
$C=C(k_1,k_2,\ldots,k_N,D) > 0$ such that
\begin{equation}\label{estimatessol}
\|\Psi\|_{\mathcal{H}}\leq C\(\|u^{in}\|_{L^2(\Gamma_1)}+\|\nabla u^{in}\|_{L^2(\Gamma_1)}\|\).
\end{equation}
\end{thm}
\begin{proof}[\bf Proof]
	We first define the operator  $\mathcal{A}_0$ by
	\[
	\begin{split}
	\mathcal{A}_0:=\diag\left(\mathcal{L}_{0,i,i-1},\mathcal{M}_{0,i}, \mathcal{L}_{0,i,i+1}\right),
	\end{split}
	\]
	where
	\[
	\begin{split}
	\mathcal{M}_{0,i}:=\begin{pmatrix}
	-\mathcal{S}_{\Gamma_i} & \mathcal{S}_{\Gamma_i} \\
	-t_{i-1}( \frac{1}{2}I+ \mathcal{K}_{\Gamma_i}^{ *}) & -\frac{1}{2}I+ \mathcal{K}_{\Gamma_i}^{*}
	\end{pmatrix},
	\end{split}
	\]
	\[
	\begin{split}
	\quad \mathcal{L}_{0,i,i+1}:=\begin{pmatrix}
	\ds \mathcal{S}_{\Gamma_{i,i+1}}& 0\\
	\nm
	\ds \mathcal{K}_{\Gamma_{i,i+1}}^{*}& 0
	\end{pmatrix},
	\end{split} \;\; \mbox{ and} \begin{split}
	\quad \mathcal{L}_{0,i,i-1}:=\begin{pmatrix}
	\ds 0&-\mathcal{S}_{\Gamma_{i,i-1}}\\
	\nm
	\ds 0&-t_{i-1}\mathcal{K}_{\Gamma_{i,i-1}}^{*}
	\end{pmatrix}.
	\end{split}
	\]
	It noted that $\mathcal{A}-\mathcal{A}_0$ is a compact operator form $\mathcal{H}$ to $\mathcal{H}_1$ \cite{Ammari2007}.
	Moreover, the operator $\mathcal{A}_0:\mathcal{H}\to\mathcal{H}_1$ is invertible. 
Therefore, by the Fredholm
alternative, it sufffces to prove that $\mathcal{A}$ is injective.
If $\mathcal{A}[\Psi] = 0$, then the function u defined by
\[
u(x) = \begin{cases}
\ds \mathcal{S}_{\Gamma_1}^{k_0} [\psi_1](x), & \quad x \in D_0,\\
\nm
\ds\mathcal{S}_{{\Gamma_j}}^{k_j} [\phi_j](x)+\mathcal{S}_{{\Gamma_{j+1}}}^{k_j} [\psi_{j+1}](x),  & \quad x \in {D_j},\; j = 1,2,\ldots,N-1,\\
\nm
\ds \mathcal{S}_{{\Gamma_N}}^{k_N} [\phi_N](x) ,  & \quad x \in {D_N},
\end{cases}
\]
is the solution to \eqref{wave_equation} with $u^{in} = 0$. By using Rellich's lemma, we have $u\equiv 0 $ in $\RR^3$.
Since $\mathcal{S}_{\Gamma_1}^{k_0} [\psi_1]$ satisfies $(\Delta + k_0^2)\mathcal{S}_{\Gamma_1}^{k_0} [\psi_1] = 0$ in $\cup_{i=1}^N D_{i}$ and $\mathcal{S}_{\Gamma_1}^{k_0} [\psi_1] = 0$ on $\Gamma_1$, and given that $k_0^2$ is not a Dirichlet eigenvalue for $-\Delta$ in $\cup_{i=1}^N D_{i}$, it follows that $\mathcal{S}_{\Gamma_1}^{k_0} [\psi_1] = 0$ in $\cup_{i=1}^N D_{i}$, and consequently in $\RR^3$. Hence, we can conclude
\[
\psi_1=\left.\frac{\p \Scal^{k_0}_{\Gamma_1}[\psi_1]}{\p \nu_1}\right|_+-\left.\frac{\p \Scal^{k_0}_{\Gamma_1}[\psi_1]}{\p \nu_1}\right|_-=0.
\]
Analogously, for $j=1,2,\ldots,N-1$, since
\[
u_{D_j}:=\mathcal{S}_{{\Gamma_j}}^{k_j} [\phi_j]+\mathcal{S}_{{\Gamma_{j+1}}}^{k_j} [\psi_{j+1}]
\]  satisfies $(\Delta + k_j^2)u_{D_j} = 0$ in $\cup_{i=j+1}^N D_{i}$ and $u_{D_j} = 0$ on $\Gamma_{j+1}$, and by the assumption  that $k_j^2$ is not a Dirichlet eigenvalue for $-\Delta$ in $\cup_{i=j+1}^N D_{i}$, then $u_{D_j} = 0$ in $\cup_{i=j+1}^N D_{i}$, and then in $\cup_{i=j}^N D_{i}$. Hence, for $j=1,2,\ldots,N-1$, we have
\[
\psi_{j+1}=\left.\frac{\p \Scal^{k_j}_{\Gamma_{j+1}}[\psi_{j+1}]}{\p \nu_{j+1}}\right|_+-\left.\frac{\p \Scal^{k_j}_{\Gamma_{j+1}}[\psi_{j+1}]}{\p \nu_{j+1}}\right|_- = \left.\frac{\p u_{D_j}}{\p \nu_{j+1}}\right|_+-\left.\frac{\p u_{D_j}}{\p \nu_{j+1}}\right|_-=0.
\]

On the other hand, for $j=1,2,\ldots,N-1$, $u_{D_j}$ satisfies $(\Delta + k_j^2)u_{D_j} = 0$ in $\cup_{i=0}^{j-1} D_{i}$ and $u_{D_j} = 0$ on $\Gamma_{j}$. It follows from \cite[Lemma 2.51]{AK_book2018} that $u_{D_j}=0$ in $\cup_{i=0}^{j-1} D_{i}$,  and then in $\cup_{i=0}^{j} D_{i}$. Hence, for $j=1,2,\ldots,N-1$, we have
\[
\phi_{j}=\left.\frac{\p \Scal^{k_j}_{\Gamma_{j}}[\phi_{j}]}{\p \nu_{j}}\right|_+-\left.\frac{\p \Scal^{k_j}_{\Gamma_{j}}[\phi_{j}]}{\p \nu_{j}}\right|_- = \left.\frac{\p u_{D_j}}{\p \nu_{j}}\right|_+-\left.\frac{\p u_{D_j}}{\p \nu_{j}}\right|_-=0.
\]
Similarly, we also have $\phi_{N} = 0.$
This completes the proof of solvability of the
integral system \eqref{integralsystem}.
The estimate \eqref{estimatessol} is a consequence of solvability and the closed graph theorem.
\end{proof}

\begin{rem}
	We remark that the boundedness and positiveness of the material parameters $\rho_j$ and $\kappa_j$, $j = 0,1,\ldots,N$, guarantee the well-posedness of \eqref{densityA&AC}--\eqref{wave_equation}, which means that \eqref{densityA&AC}--\eqref{wave_equation} has only trivial solution if $u^{in}=0$.  We seek non-trivial solutions to \eqref{densityA&AC}--\eqref{wave_equation} with $u^{in}=0$ when the parameters between layers are allowed to exhibit high contrast, i.e., when $\rho_j$ or $\kappa_j$ approach  infinity or zero, for some  $j = 0,1,\ldots,N$. This can be seen in our subsequent analysis.
\end{rem}

\section{Subwavelength resonance in multi-layer and high contrast metamaterials}\label{sec3}

In this section, we consider the subwavelength resonant phenomenon for the multi-layer and high contrast (MLHC) metamaterials. Assume that the MLHC metamaterials consist of two types of materials: the matrix material and the high-contrast material. Let $\rho$ and $\kappa$ denote the density and bulk modulus of the host matrix material, and let $\rho_\rmr$ and $\kappa_\rmr$ denote the corresponding parameters of the high-contrast material. The configuration of the considered metamaterial is characterized by for $j = 0,1,\ldots,N$,
 \begin{equation}\label{nestedcomplement}
 \rho_j
 =\left\{
 \begin{array}{ll}
 \rho_\rmr, & j \quad \mbox{is odd},\\
 \rho, & j \quad \mbox{is even},
 \end{array}
 \right.\;\; \mbox{ and }\;\;\kappa_j
 =\left\{
 \begin{array}{ll}
 \kappa_\rmr, & j \quad \mbox{is odd},\\
 \kappa, & j \quad \mbox{is even}.
 \end{array}
 \right.
 \end{equation}
Let $N_r$ denote the number of the layers of the high contrast materials, with
 \[
 N_\rmr := \left\lfloor \frac{N+1}{2}\right\rfloor.
 \]
We would like to show that the $N_r$ high contrast materials can resonate within certain frequencies.

We write
$
D_\rmr = \cup_{j=1}^{N_\rmr}D_{2j-1}
$
to signify the entire resonator-nested. The wave speeds and wavenumbers of resonators and  host matrix are given by
\begin{equation}\label{auxiliaryparameters}
v_\rmr=\sqrt{\frac{\kappa_\rmr}{\rho_\rmr}}, \quad v=\sqrt{\frac{\kappa}{\rho}}, \quad k_\rmr =\frac{\omega}{v_\rmr}, \quad k=\frac{\omega}{v}.
\end{equation}
 We introduce two dimensionless contrast parameters:
 \begin{equation}\label{contrastparameter}
 \delta=\frac{\rho_\rmr}{\rho}, \,\, \tau= \frac{k_\rmr }{k}= \frac{v}{v_\rmr} =\sqrt{\frac{\rho_\rmr \kappa}{\rho \kappa_\rmr}}.
 \end{equation}
 Assume that  $
 v = O(1)$, $v_\rmr = O(1)$,  and $\tau = O(1)$; meanwhile
 $\delta \ll 1.$
 The high-contrast assumption is the cause of the underlying system’s subwavelength resonant response and will be at the center of our subsequent analysis.

 Following our earlier discussions in \eqref{Helm-solution}--\eqref{belowinteract} and using \eqref{nestedcomplement}--\eqref{contrastparameter}, the integral equations \eqref{integralsystem} can be rewritten by
\begin{equation}\label{contrastintegralsystem}
\mathcal{A}(\omega,\delta)[\Psi]=(u^{in}, \delta \frac{\partial u^{in}}{\partial \nu_1},0,0,\ldots,0)^T.
\end{equation}
 Here the $2N$-by-$2N$ matrix type operator $\mathcal{A}$ has the block tridiagonal form \eqref{layeredintegralsystem}, where
\begin{equation}\label{contrastselfinteract}
\begin{split}
\quad \mathcal{M}_i=\left\{
\begin{array}{ll}
\begin{pmatrix}
-\mathcal{S}^{k}_{\Gamma_i} & \mathcal{S}^{k_\rmr }_{\Gamma_i} \\
-\delta( \frac{1}{2}I+ \mathcal{K}_{\Gamma_i}^{k, *}) & -\frac{1}{2}I+ \mathcal{K}_{\Gamma_i}^{k_\rmr , *}
\end{pmatrix}, & i \; \mbox{ is odd},\\
\begin{pmatrix}
-\mathcal{S}^{k_\rmr }_{\Gamma_i} & \mathcal{S}^{k}_{\Gamma_i} \\
-( \frac{1}{2}I+ \mathcal{K}_{\Gamma_i}^{k_\rmr , *}) & \delta(-\frac{1}{2}I+ \mathcal{K}_{\Gamma_i}^{k, *} )
\end{pmatrix}, & i\; \mbox{ is even},
\end{array}\right.
\end{split}
\end{equation}
\begin{equation}\label{contrastinteract}
\begin{split}
\quad \mathcal{L}_{i,i+1}=\left\{
\begin{array}{ll}\begin{pmatrix}
\ds \mathcal{S}^{k_\rmr }_{\Gamma_{i,i+1}}& 0\\
\nm
\ds \mathcal{K}_{\Gamma_{i,i+1}}^{k_\rmr , *}& 0
\end{pmatrix},& i \; \mbox{ is odd},\\
\nm
\begin{pmatrix}
\ds\mathcal{S}^{k}_{\Gamma_{i,i+1}}& 0\\
\nm
\ds \delta\mathcal{K}_{\Gamma_{i,i+1}}^{k, *}& 0
\end{pmatrix},& i \; \mbox{ is even},
\end{array}\right.
\end{split}\;\mbox{ and }\; \begin{split}
\quad \mathcal{L}_{i,i-1}:=\left\{
\begin{array}{ll}
\begin{pmatrix}
\ds 0&-\mathcal{S}^{k_\rmr }_{\Gamma_{i,i-1}}\\
\nm
\ds 0&-\mathcal{K}_{\Gamma_{i,i-1}}^{k_\rmr , *}
\end{pmatrix},& i\; \mbox{ is even},\\
\nm
\begin{pmatrix}
0&-\mathcal{S}^{k}_{\Gamma_{i,i-1}}\\
0&-\delta \mathcal{K}_{\Gamma_{i,i-1}}^{k, *}
\end{pmatrix},& i \; \mbox{ is odd}.\\
\end{array}\right.
\end{split}
\end{equation}

To comprehensively analysze the resonant behavior of the scattering system with $N_\rmr$ nested scatterers, we define the subwavelength resonant frequencies and resonant modes of the system based on the high contrast $\delta$ given in \eqref{contrastparameter}.

\begin{defn}\label{defn:resonance}
	Given $\delta>0$, a subwavelength resonant frequency (eigenfrequency) $\omega=\omega(\delta)\in\mathbb{C}$ is defined to be such that

	$\,$(i) in the case that $u^{in} = 0$, there exists a nontrivial solution to \eqref{wave_equation}, known as an associated resonant mode (eigenmode);
	
	(ii) $\omega$ depends continuously on $\delta$ and satisfies $\omega\to0$ as $\delta\to0$.
\end{defn}

According to Definition \ref{defn:resonance} and Theorem \ref{well-posedness}, finding a nontrivial solution to \eqref{wave_equation} is equivalent to finding a nontrivial function $\Psi\in \mathcal{H}$ such that
\begin{equation}\label{Awdoperaror}
\mathcal{A}(\omega,\delta)[\Psi] = 0.
\end{equation}
This can be treated as finding the characteristic value of the operator-valued analytic function $\mathcal{A}(\omega,\delta)$ with respect to $\omega$. For the reader's convenience, we provide some background on characteristic values \cite{AK_book2018}.

\begin{defn}\label{defn:char_value}
	Let $X$ and $Y$ be two Banach spaces. Let $\mathcal{U}(z_0)$ be the set of all operator-valued functions with values in $\mathcal{B}(X, Y)$  which are holomorphic in some neighborhood of $z_0$, except possibly at $z_0$. The point $z_0$ is called a characteristic value of $\mathcal{T}(z)\in\mathcal{U}(z_0)$ if there exists a vector-valued function $\phi(z)$ with values in $\mathbb{C}$ such that
	
	$\,$(i) $\phi(z)$ is holomorphic at $z_0$ and $\phi(z_0)\neq 0$;
	
	(ii)  $\mathcal{T}(z)\phi(z)$ is holomorphic at $z_0$ and $\mathcal{T}(z_0)\phi(z_0)=0$.\\		
	Here, $\phi(z)$ is called a root function of $\mathcal{T}(z)$ associated with the characteristic value $z_0$. The vector $\phi_0 = \phi(z_0)$ is called an eigenvector.
\end{defn}

Based Definitions~\ref{defn:resonance} and \ref{defn:char_value}, it is easy to see that finding a subwavelength resonant frequency of the high-contrast Helmholtz system is equivalent to finding, for a given material contrast $\delta$, a characteristic value $\omega$ of $\mathcal{A}(\omega,\delta)$ which is such that $\omega(\delta)\to0$ as $\delta\to0$.
For $\delta\ll 1$, nontrivial solutions to \eqref{Awdoperaror}  should be perturbations of the elements of  $\ker(\mathcal{A}(0,0))$.
 We will see that $\ker(\mathcal{A}(0,0))$ has dimension equal to $N_\rmr$, the number of distinct resonators in the $N$-layer structure.
 Once we understand $\ker(\mathcal{A}(0,0))$, we can investigate the characteristic values of $\mathcal{A}(\omega,\delta)$ for small $\omega$ and $\delta$ as perturbations of this space. The analysis relies  on the asymptotic perturbation and Gohberg-Sigal theory (i.e., the generalized Rouch\'e theorem to operator-valued functions) \cite{ AK_book2018}, enabling us to prove the existence of subwavelength resonances that satisfy Definition \ref{defn:resonance}.

 For the sake of convenience to the readers and self-containedness of the paper, we introduce the generalised Rouch\'e theorem  \cite{AK_book2018}. Suppose that $z_0$ is a characteristic value of the function $\mathcal{T}(z)$ and $\phi(z)$ is an associated root function. Then there exists a number $m(\phi)\geq 1$ and a vector-valued function $\psi(z)$ with values in $Y$, holomorphic at $z_0$, such that
 \[
\mathcal{T}(z)\phi(z)=(z-z_0)^{m(\phi)}\psi(z), \quad \psi(z_0)\neq 0.
 \]
 The number $m(\phi)$ is called the multiplicity of the root function $\phi(z)$. For $\phi_0\in\ker \mathcal{T}(z_0)$, we define the rank of $\phi_0$, denoted by $\mbox{rank}(\phi_0)$, to be the maximum of the multiplicities of all root functions $\phi(z)$ with $\phi(z_0)=\phi_0$.
Suppose that the space of $\ker \mathcal{T}(z_0)$ is spanned by  $\phi_0^j$, $j=1,2, \dots, m$. We call
\[
M(\mathcal{T}(z_0)):=\sum_{j=1}^m \mbox{rank}(\phi_0^j)
\]
the null multiplicity of the characteristic value $z_0$ of $\mathcal{T}(z)$.

Let $V$ be a simply connected bounded domain with rectifiable boundary $\partial V$. An operator-valued function $\mathcal{T}(z)$ which is finitely meromorphic and of Fredholm type in $V$ and continuous on $\partial V$ is called {\it normal} with respect to $\partial V$ if the operator $\mathcal{T}(z)$ is invertible in $\overline{V}$ , except for a finite number of points of $V$ which are normal points of $\mathcal{T}(z)$.

In the following, we will assume that $\mathcal{T}(z)$ has no poles in $V$. If $\mathcal{T}(z)$ is normal with respect to the $\partial V$ and $z_i$, $i = 1,2,\ldots, m$, are
all its characteristic values lying in $V$, the full multiplicity $\mathcal{M} (\mathcal{T}(z); \partial V) $
of $\mathcal{T}(z)$ in $V$ is the number of characteristic values of $\mathcal{T}(z)$ in $V$, counted with
their multiplicities.

The  generalized Rouch\'e theorem to operator-valued functions is given  in the following.

\begin{lem}[see \cite{AK_book2018}]
	\label{GSTheorem}
	Let $\mathcal{T}(z)$ be an operator-valued function which is normal with respect to $\partial V$. If an operator-valued function $\mathcal{R}(z)$, which is finitely meromorphic in $V$ and continuous on $\partial V$ satisfies the condition
	\[
	\norm{\mathcal{T}^{-1}(z) \mathcal{R}(z)}_{\mathcal{B}(X, X)}<1, \quad \forall z\in \partial V,
	\]
	then $\mathcal{T}(z) + \mathcal{R}(z)$ is also normal with respect to  $\partial V$ and
	\[
	\mathcal{M} (\mathcal{T}(z); \partial V) = \mathcal{M} (\mathcal{T}(z) + \mathcal{R}(z); \partial V),
	\]
	if there is no poles for $\mathcal{T}(z)$ in $V$.
\end{lem}

We proceed by employing Lemma \ref{GSTheorem} to establish the existence of subwavelength resonances as defined in Definition \ref{defn:resonance}. To this end, we begin by analyzing the kernel of $\mathcal{A}(0,0)$,
where the matrix operator $\mathcal{A}(0,0)$ is given by $\mathcal{A}(0,0) = \diag\left(\mathcal{L}^0_{i,i-1},\mathcal{M}^0_i, \mathcal{L}^0_{i,i+1}\right)$, with
\begin{equation}\label{contrastselfinteract0-2}
\begin{split}
\quad \mathcal{M}^0_i=\left\{
\begin{array}{ll}
\begin{pmatrix}
\ds -{\mathcal{S}}_{\Gamma_i} & {\mathcal{S}}_{\Gamma_i} \\
\nm
\ds 0 & -\frac{1}{2}I+ \mathcal{K}_{\Gamma_i}^{*}
\end{pmatrix}, & i \; \mbox{ is odd},\\
\nm
\begin{pmatrix}
\ds -{\mathcal{S}}_{\Gamma_i} & {\mathcal{S}}_{\Gamma_i} \\
\nm
\ds -\( \frac{1}{2}I+ \mathcal{K}_{\Gamma_i}^{*}\) & 0
\end{pmatrix}, & i\; \mbox{ is even},
\end{array}\right.
\end{split}
\end{equation}
	\begin{equation}\label{interact}
\begin{split}
\quad \mathcal{L}^{0}_{i,i+1}=\left\{
\begin{array}{ll}\begin{pmatrix}
\mathcal{S}_{\Gamma_{i,i+1}}& 0\\
\nm
\mathcal{K}_{\Gamma_{i,i+1}}^{ *}& 0
\end{pmatrix},& i \; \mbox{ is odd},\\
\nm
\begin{pmatrix}
\mathcal{S}_{\Gamma_{i,i+1}}& 0\\
\nm
0& 0
\end{pmatrix},& i \; \mbox{ is even},
\end{array}\right.
\end{split}\;\mbox{ and }\; \begin{split}
\quad \mathcal{L}^{0}_{i,i-1}:=\left\{
\begin{array}{ll}
\begin{pmatrix}
0&-\mathcal{S}_{\Gamma_{i,i-1}}\\
\nm
0&-\mathcal{K}_{\Gamma_{i,i-1}}^{ *}
\end{pmatrix},& i\; \mbox{ is even},\\
\nm
\begin{pmatrix}
0&-\mathcal{S}_{\Gamma_{i,i-1}}\\
\nm
0&0
\end{pmatrix},& i \; \mbox{ is odd}.
\end{array}\right.
\end{split}
\end{equation}

\begin{lem}\label{A00properties}
We have
$ \ker \(\mathcal{A}(0,0)\) = \textup{span}\,\left\{\hat\Psi_1, \hat\Psi_2,\ldots,\hat\Psi_{N_\rmr}\right\}$ where
\begin{equation}\label{KerA00}
\hat\Psi_j = (0,\ldots,0,\varphi_{2j-1},\varphi_{2j-1},0,\ldots,0)^T
\end{equation}
is the $2N$-dimensional vector with the $(4j-3)$-th and $(4j-2)$-th entrances being $\varphi_{2j-1}= \mathcal{S}_{\Gamma_{2j-1}}^{-1}[\chi_{\Gamma_{2j-1}}].$
\end{lem}
\begin{proof}[\bf Proof]
	By the invertibility of the operators ${\mathcal{S}}_{\Gamma_i}$ and $ \frac{1}{2}I+ \mathcal{K}_{\Gamma_i}^{*}$, we have that the operator $\mathcal{M}^0_i$ is invertible if $i$ is even. Moreover, We can obtain that from \cite[Lemma 2.7]{AK_book2018}, if $i$ is odd,
	\[
	\ker (\mathcal{M}^0_i) = \textup{span}\,\{(\mathcal{S}_{\Gamma_{i}}^{-1}[\chi_{\Gamma_{i}}],\mathcal{S}_{\Gamma_{i}}^{-1}[\chi_{\Gamma_{i}}])\}.
	\]
	From the expression of the operator $\mathcal{A}(0,0)$, we can conclude that this lemma holds.
\end{proof}

\begin{lem}\label{kernelrank2}
The rank of $\hat\Psi_j $ defined in \eqref{KerA00}  is 2, for $j =1,2,\ldots,N_\rmr$.	
\end{lem}
\begin{proof}[\bf Proof]
From Lemma \ref{A00properties} we see that $\mathcal{A}(0,0)$ has an $N_\rmr$-dimensional kernel. It is clear that $\omega=0$ is a characteristic value of $\mathcal{A}(\omega,0)$. In view of \eqref{series_K}, we have
\begin{equation}
\(-\frac{1}{2}I+ \mathcal{K}_{\Gamma_{2j-1}}^{k_\rmr , *}\)[\varphi_{2j-1}](x) = \omega^2 h(x,\omega),\quad x\in \Gamma_{2j-1},\; j = 1,2,\ldots,N_\rmr,
\end{equation}
for some function $h$ which is holomorphic as a function of $\omega$ in a neighborhood of 0.
By interchanging orders of integration and using \eqref{rank2}, it is easy to see that for $j = 1,2,\ldots,N_\rmr,$
\[
\begin{aligned}
&\quad \(\mathcal{K}_{\Gamma_{2j-1}, 2}[\varphi_{2j-1}],1\)_{L^2(\Gamma_{2j-1})} = \(\varphi_{2j-1},\mathcal{K}_{\Gamma_{2j-1}, 2}^*[1]\)_{L^2(\Gamma_{2j-1})} \\
=& -\int_{\Gamma_{2j-1}} \mathcal{S}_{\Gamma_{2j-1}}^{-1}[1](x) \int_{\cup_{i=2j-1}^{N_\rmr}D_i} G_0(x-y)~\d y~\d \sigma(x)\\
=& -\int_{\cup_{i=2j-1}^{N_\rmr}D_i}\int_{\Gamma_{2j-1}} G_0(x-y)\mathcal{S}_{\Gamma_{2j-1}}^{-1}[1](x)  ~\d \sigma(x)~\d y\\
=&-\sum_{i=2j-1}^{N_\rmr}Vol(D_i)\neq 0.
\end{aligned}
\]
Thus, the function $h(x, 0)$ is not identically zero. Therefore, for all $j =1,2,\ldots,N_\rmr$,  the rank of $\hat\Psi_j $ is 2.	
\end{proof}

When the material parameters are real, it follows that $\overline{\mathcal{A}(\omega, \delta)} = \mathcal{A}(-\overline{\omega}, \delta)$, indicating that the subwavelength resonant frequencies are symmetric about the imaginary axis.

\begin{lem}[see \cite{dyatlov2019mathematical}] \label{symmetric_res}
	The set of subwavelength resonant frequencies is symmetric about the imaginary axis. Specifically, if $\omega$ satisfies \eqref{Awdoperaror} for some nontrivial $\Psi \in \mathcal{H}$, then
	\begin{equation*}
	\mathcal{A}(-\overline{\omega},\delta)
	\left[\overline{\Psi}\right]
	= 0.
	\end{equation*}
\end{lem}
In view of Lemma \ref{symmetric_res}, we will henceforth present results only for subwavelength resonant frequencies with positive real parts. Using the generalized Rouché theorem in Lemma \ref{GSTheorem}, we now establish the following theorem.

\begin{thm}\label{thmres}
Consider an $N$-layer structure comprising $N_\rmr$ subwavelength nested resonators in $\mathbb{R}^3$.
	For sufficiently small $\delta>0$, there exist $N_\rmr$ subwavelength resonant frequencies $\omega_{N,1}(\delta),\dots,\omega_{N,N_\rmr}(\delta)$ with positive real parts.
\end{thm}
\begin{proof}[\bf Proof]
	It is easy to see that $\omega=0$ is a characteristic value of $\mathcal{A}(\omega,0)$ since, by Lemma \ref{A00properties}, the kernel of $\mathcal{A}(0,0)$ is spanned  by $\hat\Psi_1, \hat\Psi_2,\ldots,\hat\Psi_{N_\rmr}$.
	Then, based on the linear independence of $\hat\Psi_1, \hat\Psi_2,\ldots,\hat\Psi_{N_\rmr}$ and Lemma \ref{kernelrank2}, the multiplicity of the
	characteristic value $\omega=0$ is $2N_\rmr$.

	Since $\mathcal{A}(\omega,0)$ is a Fredholm operator, there exists a small curve $\partial V \in \mathbb{C}$ enclosing the origin such that $\mathcal{A}(\omega,0)$ is invertible for $\omega \in \partial V$, and $\omega = 0$ is the only characteristic value of $\mathcal{A}(\omega,0)$ within $V$. By \cite[Lemma 1.11]{AK_book2018}, $\mathcal{A}(\omega,0)$ is normal with respect to $\partial V$. Furthermore, the operator $\mathcal{A}(\omega,\delta) - \mathcal{A}(\omega,0)$, as a function of $\omega$, is holomorphic in $V$ and continuous on $\partial V$.
	For  $\delta\ll 1$, one has that for $\omega\in \p V$,
	\[
	\norm{\mathcal{A}(\omega,0)^{-1}(\mathcal{A}(\omega,\delta)-\mathcal{A}(\omega,0))}_{\mathcal{B}(\mathcal{H},\mathcal{H})}<1.
	\]
	By Lemma \ref{GSTheorem}, we have
	\[
	\mathcal{M} (\mathcal{A}(\omega,\delta); \partial V) = \mathcal{M} (\mathcal{A}(\omega,0); \partial V) = 2N_\rmr,
	\]
	implying the operator $\mathcal{A}(\omega,\delta)$ has
	$2N_\rmr$ characteristic values inside $V$ for small enough $\delta$.
	Therefore, by Lemma \ref{symmetric_res},
	$N_\rmr$ of these, namely, $\omega_{N,1}(\delta),\dots,\omega_{N,N_\rmr}(\delta)$
	have positive real parts, while $N_\rmr$ characteristic values have negative real parts.
	
	Similarly, we can conclude that $\omega_{N,1}(\delta),\dots,\omega_{N,N_\rmr}(\delta)$  are continuous with respect to $\delta$. Specifically, if $U \in \mathbb{C}$ is a neighborhood of $\omega_{N,j}
	(\delta_1)$, $j = 1, 2,\ldots,N_\rmr$, we can express
	\[
	\mathcal{A}(\omega,\delta_2) = \mathcal{A}(\omega,\delta_1)+ \(\mathcal{A}(\omega,\delta_2)-\mathcal{A}(\omega,\delta_1)\)
	\]
	and from Lemma \ref{GSTheorem} it follows that $\omega_{N,j}
	(\delta_2) \in U$ when $| \delta_1 - \delta_2|$ is small enough.
The proof is complete.
\end{proof}

\section{Subwavelength resonances in multi-layer concentric balls}\label{sec4}

It is known that the subwavelength resonant frequency is associated with the shape of the resonators. However, breaking the rotational symmetry of the resonators does not result in mode splitting (eigenfrequency separation). In other words, altering the shape of the resonator does not change the number of subwavelength resonant frequencies. Due to the fact, we consider the Helmholtz system \eqref{wave_equation}, where $D$ is a multi-layer concentric ball, as shown in Figures \ref{single_oscillator}--\ref{dual_oscillator}. These figures depict concentric balls with layer numbers ranging from 1 to 4. Specifically, we give a sequence of layers: $D_0, D_1, \dots, D_{N}$ by
\begin{equation}\label{eq:aj}
D_{0}:=\{r>r_{1}\}, \quad D_j:=\{r_{j+1}<r\leqslant r_{j}\}, \quad  j=1,2,\ldots, N-1, \quad D_N:=\{r\leqslant r_N\},
\end{equation}
and the interfaces between the adjacent layers can be rewritten by
\begin{equation}\label{interface}
\Gamma_j:=\left\{|x|=r_j\right\}, \quad  j=1,2,\ldots,N,
\end{equation}
where $N\in \mathbb{N}$ and $r_j\in\mathbb{R}_+$.

Let $j_n(t)$ and $h^{(1)}_n(t)$ denote the spherical Bessel and Hankel functions of the first kind of order $n$, respectively, and let $Y_n(\hat{x})$ denote the spherical harmonics. Using spherical coordinates, the solution $u$ to \eqref{wave_equation}, with material parameters specified in \eqref{nestedcomplement}--\eqref{contrastparameter}, can be expressed as
\begin{equation}\label{Helm-solution2}
u(x) = \begin{cases}
\ds \sum_{n=0}^{+\infty}a_{1,n}h^{(1)}_n(kr)Y_n, & \quad x \in D_0,\\
\nm
\ds \sum_{n=0}^{+\infty}\(b_{j,n}j_n(k_\rmr  r)Y_n + a_{j+1,n}h^{(1)}_n(k_\rmr  r)Y_n\),  & \quad x \in {D_j},\; j\; \mbox{ is odd},\\
\nm
\ds \sum_{n=0}^{+\infty}\(b_{j,n}j_n(k r)Y_n + a_{j+1,n}h^{(1)}_n(k r)Y_n\),  & \quad x \in {D_j},\; j\; \mbox{ is even},
\end{cases}
\end{equation}
where $a_{N+1,n}=0$.
By using the transmission conditions across $\Gamma_j$, $j=1,2,\ldots,N$,  the constant parameters satisfy
\[
\begin{pmatrix}
{M}_{1,n} & {R}_{1,n} & &  & & \\
{L}_{2,n} & {M}_{2,n} &{R}_{2,n}& &&\\
&{L}_{3,n} & {M}_{3,n} & {R}_{3,n} & &\\
& & \ddots &\ddots & \ddots& \\
&  & & {L}_{N-1,n} & {M}_{N-1,n} & {R}_{N-1,n}\\
&  & &  &{L}_{N,n} & {M}_{N,n}
\end{pmatrix}
\begin{pmatrix}
a_{1,n}\\b_{1,n}\\ a_{2,n}\\b_{2,n}\\ \vdots\\a_{N,n}\\b_{N,n}
\end{pmatrix} = 0,
\]
for all $n\in \mathbb{N}\cup\{0\}$,
where
\begin{equation}\label{contrastselfinteract2}
\begin{split}
\quad {M}_{i,n}=\left\{
\begin{array}{ll}
\begin{pmatrix}
-h^{(1)}_n(kr_i) & j_n(k_\rmr  r_i) \\
\nm
-\delta h^{(1)'}_n(kr_i) & \tau j'_n(k_\rmr  r_i)
\end{pmatrix}, & i \; \mbox{ is odd},\\
\nm
\begin{pmatrix}
-h^{(1)}_n(k_\rmr  r_i) & j_n(k r_i) \\
\nm
-\tau  h^{(1)'}_n(k_\rmr  r_i) & \delta j'_n(k r_i)
\end{pmatrix}, & i\; \mbox{ is even},
\end{array}\right.
\end{split}
\end{equation}
\begin{equation}\label{contrastinteract2}
\begin{split}
\quad R_{i,n}=\left\{
\begin{array}{ll}\begin{pmatrix}
h^{(1)}_n(k_\rmr  r_i)& 0\\
\nm
\tau h^{(1)'}_n(k_\rmr  r_i)& 0
\end{pmatrix},& i \; \mbox{ is odd},\\
\nm
\begin{pmatrix}
h^{(1)}_n(k r_i)& 0\\
\nm
\delta h^{(1)'}_n(k r_i)& 0
\end{pmatrix},& i \; \mbox{ is even},
\end{array}\right.
\end{split}\;\mbox{ and }\; \begin{split}
\quad {L}_{i,n}:=\left\{
\begin{array}{ll}
\begin{pmatrix}
0&-j_n(k_\rmr r_i)\\
\nm
0&-\tau j'_n(k_\rmr r_i)
\end{pmatrix},& i\; \mbox{ is even},\\
\nm
\begin{pmatrix}
0&-j_n(kr_i)\\
\nm
0&-\delta j'_n(kr_i)
\end{pmatrix},& i \; \mbox{ is odd}.\\
\end{array}\right.
\end{split}
\end{equation}

It is important to emphasize that, from a physical perspective, we are concerned with the resonance of nested materials, which corresponds to the system's lowest resonant frequency. At this frequency, the system exhibits a factor corresponding to $n = 0$, as the lowest resonance is characterized by the fewest number of oscillations.
Consequently,  the matrix
\begin{equation}\label{MLCBM}
\bm{A}_N=\bm{A}_N(\omega,\delta):=\diag\left(L_{i,0},M_{i,0}, R_{i,0}\right)
\end{equation}
becomes singular.

From Theorem \ref{thmres}, we know that $\det(\bm{A}_N) = 0$ process $N_\rmr := \left\lfloor \frac{N+1}{2}\right\rfloor$ eigenfrequencies  with positive real parts.
This can be seen in the proofs of Theorems \ref{single_reson_frequency}--\ref{dual_reson_frequency}, i.e., the primary reason for mode splitting lies in the fact that as the number of nested resonators increases, the degree of the corresponding characteristic polynomial also increases, while the type of resonance (which consists solely of monopolar resonances) remains unchanged.
Thus, by Galois’ theory \cite{JPT2001}, it is possible to obtain the exact formulas  for the eigenvalues of structures with layer numbers $\leq 8$ by using root-finding formulas for quadratic, cubic, and quartic equations. In what follows, to simplify the calculations and avoid the complexity of cubic or quartic formulas, we shall drive the exact formulas  for the eigenfrequencies for single-resonator, dual-resonator models in the rest of this section.
For structures with a large number of layers, we shall provide  numerical computations of the  eigenfrequencies in the next section, which is also important from a practical perspective.
The following asymptotic expansions shall be used
\begin{equation}\label{besssmall}
j_0(t) =  1 - \frac{t^2}{6} + \frac{t^4}{120} + O(t^6),
\end{equation}
\begin{equation}\label{hankelsmall}
h^{(1)}_0(t) = 1 - \frac{t^2}{6} + \frac{t^4}{120} + \mathrm{i}(-\frac{1}{t} + \frac{t}{2} - \frac{t^3}{24})+ O(t^5),
\end{equation}
for $t\ll 1$.

\subsection{Single-resonator}
In this subsection, we shall derive the eigenfrequency of a single-resonator.
We now implement the single-resonator  to two simple systems, a solid sphere (see, Figure \ref{single_oscillator} (a))  and  a spherical shell (see, Figure \ref{single_oscillator} (b)).
\begin{figure}[!htpb]
	\centering
	\subfigure[Solid sphere]{\includegraphics[width=0.32\textwidth]{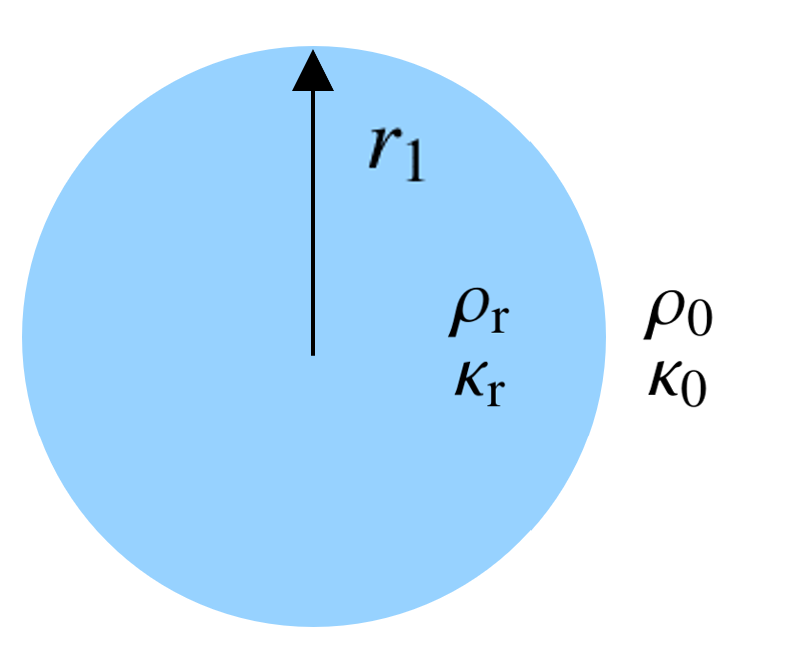}}\quad
	\subfigure[Spherical shell]{\includegraphics[width=0.31\textwidth]{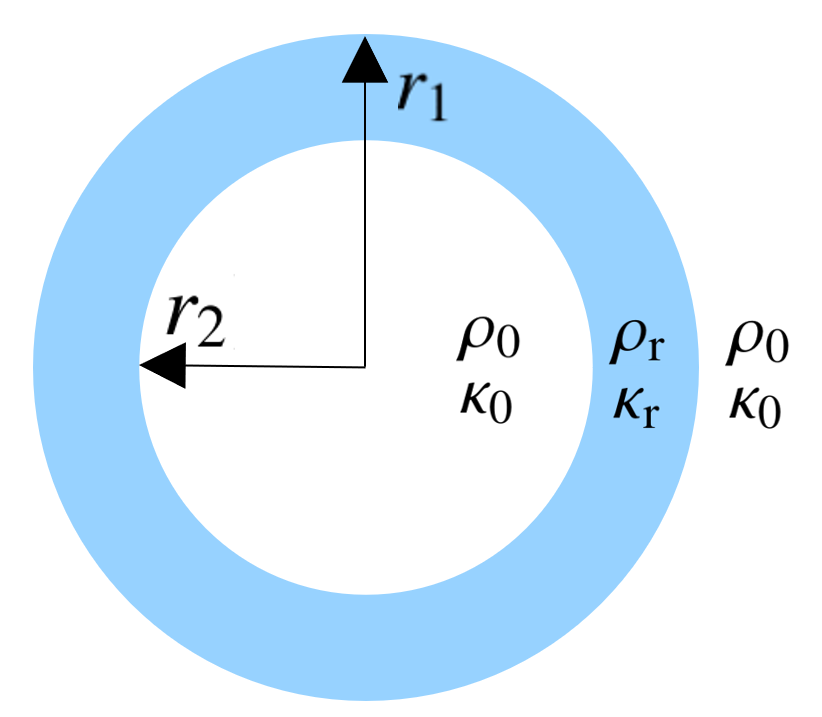}}
	\caption{\label{fig34} Schematic illustration of single-resonator (a) $N=1$, (b) $N=2$. }\label{single_oscillator}
\end{figure}

\begin{thm}\label{single_reson_frequency}
	Let the single-resonator be presented by Figure \ref{fig34}, where the material parameters are given in \eqref{nestedcomplement}--\eqref{contrastparameter}. Then, in the quasi-static regime, there exists one subwavelength resonant frequency for monopolar mode $n=0$.
	Specifically, we have that
	
	$\;\bf(i)$ for $N=1$, as shown in Figure \ref{fig34} (a),  the subwavelength resonant frequency is expressed by
	\begin{equation}\label{solidfreqc}
	\omega_{1,1}(\delta) = \frac{\sqrt{3 }v_\rmr}{{r_1}}\delta^{\frac{1}{2}} - \mathrm{i}\frac{3v_\rmr}{2\tau r_1}\delta +   O (\delta^{\frac{3}{2}});
	\end{equation}
	
	$\bf(ii)$ for $N=2$, as shown in Figure \ref{fig34} (b), the subwavelength resonant frequency is expressed by
	\begin{equation}\label{shellfreqc}
		\omega_{2,1}(\delta) = \frac{\sqrt{3 r_1}v_\rmr}{\sqrt{r_1^3-r_2^3}}\delta^{\frac{1}{2}} - \mathrm{i}\frac{3r_1^2v_\rmr}{2\tau  \big(r_1^3-r_2^3\big)}\delta +   O (\delta^{\frac{3}{2}}).
	\end{equation}
\end{thm}
\begin{proof}[\bf Proof]
	$\bf(i)$
	Plugging the formulas \eqref{besssmall}--\eqref{hankelsmall} into $\det (\bm{A}_1) = 0$, we have
\[
\mathrm{i}\frac{ \tau^2v_\rmr^2}{{r_{1}}^2}\delta-\i \frac{\tau^2 }{3}\omega^2+\frac{r_{1}\tau }{3v_{r}}\omega^3+O(\omega^4) + O(\delta^2) = 0.
\]
	It can be seen that $\delta = O(\omega^2)$, and thus $\omega_1(\delta) = O(\sqrt{\delta})$.
	We express the $\omega_1(\delta)$ in the following asymptotic expansion
	$$
	\omega_1(\delta) = a_1 \delta^{\frac{1}{2}} + a_2 \delta + O (\delta^{\frac{3}{2}}).
	$$
	Thus, we obtain
	\[
	\begin{aligned}
	&\mathrm{i}\frac{ \tau^2v_\rmr^2}{{r_{1}}^2}\delta-\i \frac{\tau^2 }{3}\(a_1 \delta^{\frac{1}{2}} + a_2 \delta + O (\delta^{\frac{3}{2}})\)^2\\
	&+\frac{r_{1}\tau }{3v_{r}}\(a_1 \delta^{\frac{1}{2}} + a_2 \delta + O (\delta^{\frac{3}{2}})\)^3 +O(\omega^4) + O(\delta^2) = 0.
	\end{aligned}
	\]
	From the coefficients of the $\delta$ and $\delta^{\frac{3}{2}}$ terms, we have
	\[
	\begin{aligned}
	\mathrm{i}\frac{ \tau^2v_\rmr^2}{{r_{1}}^2}-\i \frac{\tau^2 }{3}a_1^2 = 0,\quad \mbox{ and } \quad- \mathrm{i}\frac{2\tau^2}{3}a_1a_2+\frac{r_{1}\tau }{3v_{r}} a_1^3  = 0,
	\end{aligned}
	\]
	which further gives
	\[
	\begin{aligned}
	a_1^2 =  \frac{3v_\rmr^2}{r_1^2},\quad a_2 = - \mathrm{i}\frac{3v_\rmr}{2\tau r_1}.
	\end{aligned}
	\]
	
$\bf(ii)$ Similarly, it follows from $\det (\bm{A}_2) = 0$ that	
	\begin{equation}
	-\frac{\tau^3 {v_{r}}^4}{{r_{1}}^2{r_{2}}^2}\delta+
	\frac{\,\left({r_{1}}^3-{r_{2}}^3\right)\tau ^3 v^2_{r}}{3\,{r_{1}}^3\,{r_{2}}^2}\omega^2 +
	\i \frac{\tau ^2\,v_{r}\,\left({r_{1}}^3-{r_{2}}^3\right)}{3\,{r_{1}}^2\,{r_{2}}^2}\omega^3+O(\omega^4) + O(\delta^2) = 0.
	\end{equation}
	There also have $\delta = O(\omega^2)$, and thus $\omega_2(\delta) = O(\sqrt{\delta})$.  We write the $\omega_2(\delta)$ in the following asymptotic expansion
	\[
	\omega_2(\delta) = a_1 \delta^{\frac{1}{2}} + a_2 \delta + O (\delta^{\frac{3}{2}}).
	\]
	In a similar manner, we have
		\[
	\begin{aligned}
	a_1^2 =  \frac{{3 r_1}v^2_r}{{r_1^3-r_2^3}},\quad a_2 = - \mathrm{i}\frac{3r_1^2v_\rmr}{2\tau  \big(r_1^3-r_2^3\big)}.
	\end{aligned}
	\]
	The proof is complete.
\end{proof}

\begin{rem}
	(i)
	For single-layer resonators of general shape, the eigenfrequencies have been already formulated in \cite{AFGLZ_AIHPCAN}.
	Following the methods in \cite{AFGLZ_AIHPCAN}, we can represent the  eigenfrequency of a single-resonator with a general shape in a unified way:
		\[
	\omega_{N,1}(\delta) = v_\rmr\sqrt{ \frac{  Cap(\Gamma_1)}{Vol(D_1)}}\delta^{\frac{1}{2}}-\i\frac{ Cap(\Gamma_1)^2 v_\rmr}{8\pi  \tau Vol(D_1) }\delta + O(\delta^{\frac{3}{2}})
	\]	
	for $N=1,2$, where $Vol(D_1)$ is the volume of $D_1$ and $Cap(\Gamma_1):= - (\mathcal{S}_{\Gamma_1}^{-1}[\chi_{ \Gamma_1}], \chi_{\Gamma_1})$ is the capacity of $\Gamma_1$.
	In the concentric ball case, we have $Cap(\Gamma_1) = 4\pi r_1$ for $N=1,2$, $Vol(D_1) = \frac{4\pi r_1^3}{3}$ for $N=1$, and $Vol(D_1) = \frac{4\pi \big(r_1^3-r_2^3\big)}{3}$ for $N=2$. Therefore, in this way, we can also obtain \eqref{solidfreqc} and \eqref{shellfreqc}.
	
	(ii) In the limit $r_2\to 0$, one has
$\omega_{2,1}\searrow  \omega_{1,1}$. while in the limit $r_1\to+\infty$, one has $\omega_{2,1}\to 0$.
This indicates that, in the scenario of a cavity situated within a bulk resonator, subwavelength resonance cannot occur. In fact, the corresponding  integral operator is invertible in this case.
This distinction is fundamentally why the mode splitting observed in this case differs from that in multi-layer plasmonic materials\cite{FangdengMMA23,FDLMMA15,Ammari2013}. In the latter, plasma exciton modes interact and hybridize, resulting in the excitation of localized surface plasma excitons at each interface \cite{PRHN2003SCI,KZDF}. In contrast, the resonance in the former case is confined solely to the outer surface.

 (iii)
For spherical shell, the subwavelength resonant frequency  \eqref{shellfreqc} possesses higher tunability. It can be treated as an interaction between the capacity of the outer surface and the volume of the shell. This interaction results in mode shift: the mode $\Re \omega_{2,1}$ shows an upward frequency shift (blueshift). The strength of the interaction can be seen as a function of the ratio $\frac{  Cap(\Gamma_1)}{Vol(D_1)}$, which in a concentric shell is equivalently $\frac{r_1}{r_1^3-r_2^3}$ and is referred to as Cap-to-Vol ratio (abbreviated as CVR). As the CVR increases, indicating a stronger coupling (interaction), a larger blueshift is observed.

\end{rem}
\subsection{Dual-resonator}
In this subsection, we will focus on dual-resonators with a Matryoshka-like structure. These structures consist of a solid sphere (see Figure \ref{dual_oscillator} (a)) or a spherical shell (see Figure \ref{dual_oscillator} (b)) encapsulated within an additional concentric spherical shell.

\begin{figure}[!htpb]
	\centering
	\subfigure[Three-layer structure]{\includegraphics[width=0.33\textwidth]{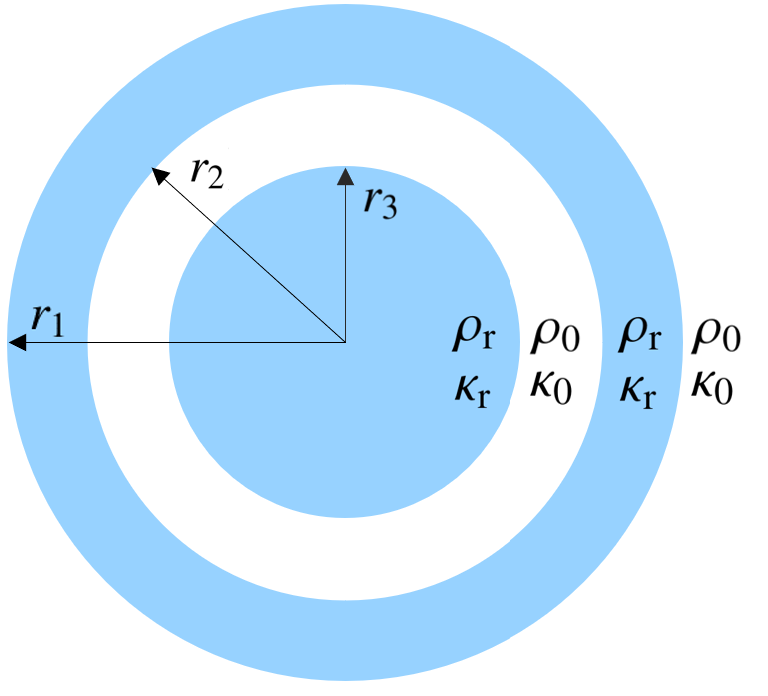}}\quad
	\subfigure[Four-layer structure]{\includegraphics[width=0.34\textwidth]{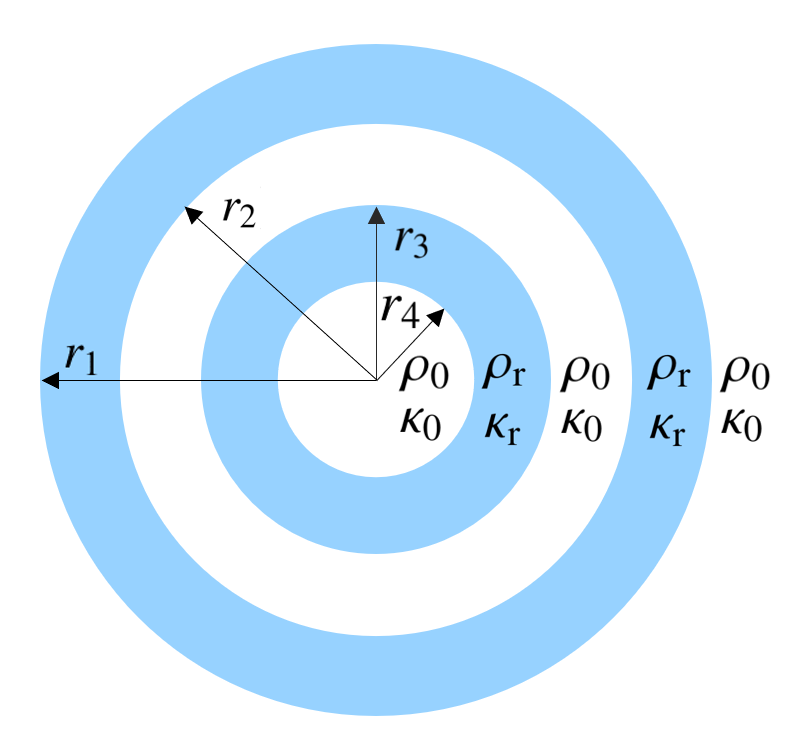}}
	\caption{\label{fig33} Schematic illustration of dual-resonator of Matryoshka-like (a) $N=3$, (b) $N=4$. }\label{dual_oscillator}
\end{figure}

\begin{thm}\label{dual_reson_frequency}
	Let the dual-resonator be presented by Figure \ref{fig33}, where the material parameters are given in \eqref{nestedcomplement}--\eqref{contrastparameter}. Then, in the quasi-static regime, there exist two subwavelength resonant frequencies for the monopolar mode $n=0$.
	Particularly, we have that
	
	$\;\bf(i)$ for $N=3$, as shown in Figure \ref{fig33} (a),  two subwavelength resonant frequencies are expressed by
	\[
	\begin{aligned}
	\omega_{3,1}(\delta)& = v_{r}\sqrt{\frac{3\,{\Xi_3}- 3 \sqrt{{\,\Xi_3^2-4{{r_{1}}\,{r_{2}}\,r_{3}^2}\,{(r_2-r_3)\big(r_1^3-r_2^3\big)}}}}{2\,{(r_2-r_3)\big(r_1^3-r_2^3\big)r_3^2}}}\delta^{\frac{1}{2}}\\
	&\quad -\i \frac{  {3\(\,-{ \Xi_3}+\sqrt{ \Xi_3^2-4{{r_{1}}\,{r_{2}}\,r_{3}^2}\,{(r_2-r_3)\big(r_1^3-r_2^3\big)}}\)}+ {6{r_{2}}\big(r_1^3-r_2^3\big)}}{ 4\tau {\big(r_1^3-r_2^3\big)}
		\sqrt{\Xi_3^2-{4{r_{1}}\,{r_{2}}\,r_{3}^2(r_2-r_3)\big(r_1^3-r_2^3\big)}   }}r_1^2 v_{r}\delta +O(\delta^{\frac{3}{2}}),
	\end{aligned}
	\]
	\[
	\begin{aligned}
	\omega_{3,2}(\delta)& = v_{r}\sqrt{\frac{3\,{\Xi_3} + 3 \sqrt{{\,\Xi_3^2-4{{r_{1}}\,{r_{2}}\,r_{3}^2}\,{(r_2-r_3)\big(r_1^3-r_2^3\big)}}}}{2\,{(r_2-r_3)\big(r_1^3-r_2^3\big)r_3^2}}}\delta^{\frac{1}{2}}\\
	&\quad -\i \frac{  {3\(\,{ \Xi_3}+\sqrt{ \Xi_3^2-4{{r_{1}}\,{r_{2}}\,r_{3}^2}\,{(r_2-r_3)\big(r_1^3-r_2^3\big)}}\)}- {6{r_{2}}\big(r_1^3-r_2^3\big)}}{ 4\tau {\big(r_1^3-r_2^3\big)}
		\sqrt{\Xi_3^2-{4{r_{1}}\,{r_{2}}\,r_{3}^2(r_2-r_3)\big(r_1^3-r_2^3\big)}   }}r_1^2 v_{r}\delta +O(\delta^{\frac{3}{2}}),
	\end{aligned}
	\]

	where $\Xi_3 =  	
	r_{2}\big(r_{1}^3-r_{2}^3+r_{3}^3\big)+r_{1}r_{3}^2(r_{2}-r_{3});$
	
	$\bf(ii)$ for $N=4$, as shown in Figure \ref{fig33} (b), two subwavelength resonant frequencies are expressed by
	\begin{equation}\label{4Lw1}
	\begin{aligned}
	\omega_{4,1}(\delta)& = v_{r}\sqrt{\frac{3\Xi_4- 3 \sqrt{\Xi_4^2-4r_{1}r_{2}r_{3} {\big(r_1^3-r_2^3\big)\big(r_3^3-r_4^3\big)\left(r_{2}-r_{3}\right)}}}{2{\big(r_1^3-r_2^3\big)\big(r_3^3-r_4^3\big)\left(r_{2}-r_{3}\right)}}}\delta^{\frac{1}{2}}\\
	&\quad -\i \frac{3 {\(-\Xi_4+  \sqrt{\Xi_4^2-4r_{1}r_{2}r_{3} {\big(r_1^3-r_2^3\big)\big(r_3^3-r_4^3\big)\left(r_{2}-r_{3}\right)}}\)}+{6r_2r_3\big(r_1^3-r_2^3\big)}}{ 4{\tau\big(r_1^3-r_2^3\big)}\sqrt{{\Xi_4}^2-{4r_1r_2r_3\big(r_1^3-r_2^3\big)\big(r_3^3-r_4^3\big)\left(r_{2}-r_{3}\right)}}}{v_{r}}r_1^2  \delta + O(\delta^{\frac{3}{2}}),
	\end{aligned}
	\end{equation}
	\begin{equation}\label{4Lw2}
	\begin{aligned}
	\omega_{4,2}(\delta) &= v_{r}\sqrt{\frac{3\Xi_4+ 3 \sqrt{\Xi_4^2-4r_{1}r_{2}r_{3} {\big(r_1^3-r_2^3\big)\big(r_3^3-r_4^3\big)\left(r_{2}-r_{3}\right)}}}{2{\big(r_1^3-r_2^3\big)\big(r_3^3-r_4^3\big)\left(r_{2}-r_{3}\right)}}}\delta^{\frac{1}{2}}\\
	&\quad -\i \frac{3 {\(\Xi_4+  \sqrt{\Xi_4^2-4r_{1}r_{2}r_{3} {\big(r_1^3-r_2^3\big)\big(r_3^3-r_4^3\big)\left(r_{2}-r_{3}\right)}}\)}-{6r_2r_3\big(r_1^3-r_2^3\big)}}{ 4{\tau\big(r_1^3-r_2^3\big)}\sqrt{{\Xi_4}^2-{4r_1r_2r_3\big(r_1^3-r_2^3\big)\big(r_3^3-r_4^3\big)\left(r_{2}-r_{3}\right)}}}{v_{r}}r_1^2 \delta + O(\delta^{\frac{3}{2}}),
	\end{aligned}
	\end{equation}
	where $\Xi_4 = r_{2}r_{3}\big(r_{1}^3-r_{2}^3+r_{3}^3-r_{4}^3\big)+r_1(r_2-r_3)\big(r_3^3-r_4^3\big)$.
\end{thm}
\begin{proof}[\bf Proof]
	$\bf(i)$
Through straightforward calculations (but rather lengthy and tedious) and from the asymptotic expansion for $\omega\ll 1$
and $\delta\ll 1$,  one has that from $\det (\bm{A}_3)=0$,
\[
\begin{aligned}
0& = -\i \frac{\tau^5 v_{r}^6}{r_{1}^2 r_{2}^2 r_{3}^2}\delta ^2
-\mathrm{i}\frac{(r_2-r_3)\big(r_1^3-r_2^3\big)\tau ^5v_{r}^2}{9r_{1}^3r_{2}^3} \omega^4+\mathrm{i}\frac{ \left(	
	r_{2}\big(r_{1}^3-r_{2}^3+r_{3}^3\big)+r_{1}r_{3}^2(r_{2}-r_{3})\right)\tau ^5v^4_{r}}{3r_{1}^3 r_{2}^3 r_{3}^2} \omega^2\delta \\
&\quad- \frac{ \big(r_{1}^3-r_{2}^3+r_{3}^3\big)\tau^4 v_{r}^3}{3r_{1}^2r_{2}^2r_{3}^2}\omega^3 \delta + \frac{\big(r_1^3-r_2^3\big)\left(r_{2}-r_{3}\right)\tau ^4v_{r}}{9r_{1}^2r_{2}^3}\omega^5+O(\omega^6) +O(\delta^3)\\
&:=-\mathrm{i}c_1\delta^2-\mathrm{i}c_2 \omega^4 +\mathrm{i}c_3 \omega^2 \delta-c_4\omega^3 \delta +c_5\omega^5+O(\omega^6)  +O(\delta^3),
\end{aligned}
\]
where
\[
c_1 = \frac{\,\tau^5 \,v_{r}^6\,}{r_{1}^2\,r_{2}^2\,r_{3}^2},\quad c_2 = \frac{(r_2-r_3)\big(r_1^3-r_2^3\big)\tau ^5{v_{r}}^2}{9r_{1}^3r_{2}^3},
\]
\[
c_3 = \frac{\( 	
	r_{2}\big(r_{1}^3-r_{2}^3+r_{3}^3\big)+r_{1}\,r_{3}^2(\,r_{2}-\,r_{3})\)\tau ^5\,v^4_{r}}{3r_{1}^3r_{2}^3r_{3}^2}:= \frac{\Xi_3\tau ^5\,v^4_{r}}{3r_{1}^3r_{2}^3r_{3}^2},
\]
\[
 c_4 = \frac{ \big(r_{1}^3-r_{2}^3+r_{3}^3\big)\tau^4 v_{r}^3}{3\,r_{1}^2\,r_{2}^2\,r_{3}^2},\;\mbox{ and }\; c_5 = \frac{\big(r_1^3-r_2^3\big)\left(r_{2}-r_{3}\right)\tau ^4v_{r}}{9r_{1}^2r_{2}^3}.
\]
It follows from $r_1>r_2>r_3$ that $c_j>0$, for $j = 1,2,3,4,5.$
It is clear that $\delta = O(\omega^2)$, and thus $\omega_3(\delta) = O(\sqrt{\delta})$.
We write the $\omega_3(\delta)$ in the following asymptotic expansion
$$
\omega_3(\delta) = a_1 \delta^{\frac{1}{2}} + a_2 \delta + O (\delta^{\frac{3}{2}}).
$$
We get
\begin{equation}
\begin{aligned}
&-\i c_1\delta^2- \i c_2 \(a_1 \delta^{\frac{1}{2}} + a_2 \delta + O (\delta^{\frac{3}{2}})\)^4 + \i c_3 \(a_1 \delta^{\frac{1}{2}} + a_2 \delta + O (\delta^{\frac{3}{2}})\)^2 \delta\\
&-c_4\(a_1 \delta^{\frac{1}{2}} + a_2 \delta + O (\delta^{\frac{3}{2}})\)^3 \delta +c_5\(a_1 \delta^{\frac{1}{2}} + a_2 \delta + O (\delta^{\frac{3}{2}})\)^5   +O(\delta^3)=0.
\end{aligned}
\end{equation}
From the coefficients of the $\delta^2$ and $\delta^{\frac{5}{2}}$ terms, we obtain
\begin{equation}\label{realpart}
 -c_1-c_2a_1^4+c_3a_1^2 =0,
\end{equation}
\begin{equation}\label{imagpart}
 -4\i c_2a_1^3a_2+2\i c_3a_1a_2-c_4a_1^3 +c_5a_1^5=0.
\end{equation}
By \eqref{realpart}, we have
\begin{equation}\label{realsol}
\begin{aligned}
a_1^2 &=   \frac{-c_3\pm\sqrt{c_3^2-4c_1c_2}}{-2c_2}\\
&= 3v^2_{r}  \frac{\,{ \Xi_3}\mp\sqrt{ \Xi_3^2-4{{r_{1}}\,{r_{2}}\,r_{3}^2}\,{(r_2-r_3)\big(r_1^3-r_2^3\big)}}}{2{(r_2-r_3)\big(r_1^3-r_2^3\big)r_3^2}}>0,
\end{aligned}
\end{equation}
where we used the fact that
\[
c_3^2-4c_1c_2=\tau^{10}v_\rmr^8\frac{\,\left(	
	r_{2}\big(r_{1}^3-r_{2}^3\big)\,-r_{1}\,r_{3}^2(\,r_{2}-\,r_{3})+r_{2}\,r_{3}^3\right)^2+4 r_{1}\,r_{2}\,{r_{3}}^5(\,r_{2}-\,r_{3})}{9{r_{1}}^6\;{r_{2}}^6\,{r_{3}}^4}>0.
\]
On the other hand, by \eqref{realpart}--\eqref{realsol}, we can obtain that
\begin{equation}
\begin{aligned}
a_2&= -\i \frac{c_4a_1^2-c_5a_1^4}{2c_3-4c_2a_1^2}
= -\i \frac{\(c_2c_4-c_3c_5\)a_1^2+c_1c_5}{\pm 2c_2\sqrt{c_3^2-4c_1c_2}}\\
&= -\i \frac{\(-r_1^2r_{3}^2\big(r_1^3-r_2^3\big)\,\left(r_{2}-r_{3}\right)^2\)a_1^2+ {3r_{1}^2\,{r_{2}}{v_{r}}^2\big(r_1^3-r_2^3\big)\,\left(r_{2}-r_{3}\right)}}{\pm 2\tau {v_{r}}{(r_2-r_3)(r_1^3-r_2^3)}	\sqrt{\Xi_3^2-{4{r_{1}}\,{r_{2}}\,r_{3}^2(r_2-r_3)(r_1^3-r_2^3)}   }}\\
&= - \i \frac{  {3\(\,-{ \Xi_3}\pm\sqrt{ \Xi_3^2-4{{r_{1}}\,{r_{2}}\,r_{3}^2}\,{(r_2-r_3)(r_1^3-r_2^3)}}\)}+ {6{r_{2}}\big(r_1^3-r_2^3\big)}}{\pm 4\tau {(r_1^3-r_2^3)}
	\sqrt{\Xi_3^2-{4{r_{1}}\,{r_{2}}\,r_{3}^2(r_2-r_3)(r_1^3-r_2^3)}   }}r_1^2 v_{r},
\end{aligned}
\end{equation}
which completes the proof for the case $N=3$.

$\bf(ii)$ In a similar manner, by rather lengthy and tedious calculations, it follows from $\det(\bm{A}_4)=0$, for for $\omega\ll 1$
and $\delta\ll 1$, that
\[
\begin{aligned}
0 & = \frac{\tau ^6v_{r}^8}{r_{1}^2\,r_{2}^2\,r_{3}^2\,r_{4}^2}\delta ^2 +\frac{\big(r_1^3-r_2^3\big)\big(r_3^3-r_4^3\big)\left(r_{2}-r_{3}\right)\tau ^6v_{r}^4}{9r_{1}^3r_{2}^3r_{3}^3r_{4}^2}
\omega^4\\
&\quad -\frac{\(r_{2}\,r_{3}\big(r_{1}^3-r_{2}^3+r_{3}^3-r_{4}^3\big)+r_1(r_2-r_3)\big(r_3^3-r_4^3\big)\)\tau ^6\,v_{r}^6}{3\,r_{1}^3\,r_{2}^3\,r_{3}^3\,r_{4}^2}\omega^2\delta\\
&\quad -\mathrm{i}\frac{\big(r_{1}^3-r_{2}^3+r_{3}^3-r_{4}^3\big)\tau ^5v_{r}^5}{3\,r_{1}^2\,r_{2}^2\,r_{3}^2\,r_{4}^2}
\omega^3\delta +\i  \frac{\big(r_1^3-r_2^3\big)\big(r_3^3-r_4^3\big)\,\left(r_{2}-r_{3}\right)\tau ^5{v_{r}}^3}{9\,r_{1}^2\,r_{2}^3\,r_{3}^3\,r_{4}^2}\omega^5 +O(\omega^6)+O(\delta^3)\\
&:=c_1\delta^2+c_2 \omega^4 -c_3 \omega^2 \delta-\mathrm{i}c_4\omega^3 \delta +\i c_5 \omega^5+O(\omega^6)  +O(\delta^3),
\end{aligned}
\]
where
\[
c_1 =  \frac{\,\tau ^6\,v_{r}^8}{r_{1}^2\,r_{2}^2\,r_{3}^2\,r_{4}^2},\quad c_2 = \frac{\big(r_1^3-r_2^3\big)\,\big(r_3^3-r_4^3\big)\,\left(r_{2}-r_{3}\right)\tau ^6v_{r}^4}{9\,r_{1}^3\,r_{2}^3\,r_{3}^3\,r_{4}^2},
\]
\[
c_3 = \frac{r_{2}\,r_{3}\big(r_{1}^3-r_{2}^3+r_{3}^3-r_{4}^3\big)+r_1(r_2-r_3)\big(r_3^3-r_4^3\big)}{3\,r_{1}^3\,r_{2}^3\,r_{3}^3\,r_{4}^2}\tau ^6v_{r}^6:=\frac{\Xi_4 \tau ^6v_{r}^6}{3\,r_{1}^3\,r_{2}^3\,r_{3}^3\,r_{4}^2},
\]
\[
c_4 = \frac{\big(r_{1}^3-r_{2}^3+r_{3}^3-r_{4}^3\big)\tau ^5v_{r}^5}{3\,r_{1}^2\,r_{2}^2\,r_{3}^2\,r_{4}^2}, \; \mbox{ and }\;
c_5 = \frac{\big(r_1^3-r_2^3\big)\big(r_3^3-r_4^3\big)\left(r_{2}-r_{3}\right)\tau ^5{v_{r}}^3}{9r_{1}^2r_{2}^3r_{3}^3r_{4}^2}.
\]
From $r_1>r_2>r_3>r_4$, we also have that $c_j>0$, for $j = 1,2,3,4,5.$
It is also clear that $\delta = O(\omega^2)$, and thus $\omega_4(\delta) = O(\sqrt{\delta})$.
We write the $\omega_4(\delta)$ in the following asymptotic expansion
$$
\omega_4(\delta) = a_1 \delta^{\frac{1}{2}} + a_2 \delta + O (\delta^{\frac{3}{2}}).
$$
We get
\begin{equation}
\begin{aligned}
&c_1\delta^2+c_2 \(a_1 \delta^{\frac{1}{2}} + a_2 \delta + O (\delta^{\frac{3}{2}})\)^4 - c_3 \(a_1 \delta^{\frac{1}{2}} + a_2 \delta + O (\delta^{\frac{3}{2}})\)^2 \delta\\
&-\i c_4\(a_1 \delta^{\frac{1}{2}} + a_2 \delta + O (\delta^{\frac{3}{2}})\)^3 \delta +\i c_5\(a_1 \delta^{\frac{1}{2}} + a_2 \delta + O (\delta^{\frac{3}{2}})\)^5  +O(\delta^3)=0.
\end{aligned}
\end{equation}
From the coefficients of the $\delta^2$ and $\delta^{\frac{5}{2}}$ terms, we obtain
\[
\begin{aligned}
& c_1^2+c_2a_1^4-c_3a_1^2 =0, \\
& 4c_2a_1^3a_2-2c_3a_1a_2-\i c_4a_1^3 + \i c_5a_1^5=0.
\end{aligned}
\]
Hence, we have
\begin{equation}
\begin{aligned}
a_1^2 &=   \frac{c_3\pm\sqrt{c_3^2-4c_1c_2}}{2c_2}\\
&=   {v_{r}}^2\frac{3\Xi_4\pm 3 \sqrt{\Xi_4^2-4r_{1}r_{2}r_{3} {\big(r_1^3-r_2^3\big)\big(r_3^3-r_4^3\big)\left(r_{2}-r_{3}\right)}}}{2{\big(r_1^3-r_2^3\big)\big(r_3^3-r_4^3\big)\left(r_{2}-r_{3}\right)}}>0,
\end{aligned}
\end{equation}
where we used the fact that
\[
c_3^2-4c_1c_2 = \tau ^{12}\,{v_{r}}^{12}\frac{\(r_{2}\,r_{3}\big(r_{1}^3-r_{2}^3+r_{3}^3-r_{4}^3\big)-r_1(r_2-r_3)(r_3^3-r_4^3)\)^2+4r_1 r_2r_3(r_2-r_3)(r_3^3-r_4^3)^2}{9\,{r_{1}}^6\,{r_{2}}^6\,{r_{3}}^6\,{r_{4}}^4}>0.
\]
Moreover, we obtain
\begin{equation}
\begin{aligned}
a_2&= -\i \frac{c_4a_1^2-c_5a_1^4}{2c_3-4c_2a_1^2}
= -\i \frac{\(c_2c_4-c_3c_5\)a_1^2+c_1c_5}{\mp 2c_2\sqrt{c_3^2-4c_1c_2}}\\
&= -\i \frac{\(-r_1^2\big(r_1^3-r_2^3\big)\,\big(r_3^3-r_4^3\big)^2\,\left(r_{2}-r_{3}\right)^2\)a_1^2+{3r_1^2r_2r_3{v_{r}}^{2}\big(r_1^3-r_2^3\big)\,\big(r_3^3-r_4^3\big)\,\left(r_{2}-r_{3}\right)}}{\mp 2{\tau\,{v_{r}}\big(r_1^3-r_2^3\big)\,\big(r_3^3-r_4^3\big)\,\left(r_{2}-r_{3}\right)}\sqrt{{\Xi_4}^2-{4r_1r_2r_3\big(r_1^3-r_2^3\big)\,\big(r_3^3-r_4^3\big)\,\left(r_{2}-r_{3}\right)}}}\\
&= -\i \frac{3 {\(-\Xi_4\mp  \sqrt{\Xi_4^2-4r_{1}r_{2}r_{3} {\,\big(r_1^3-r_2^3\big)\,\big(r_3^3-r_4^3\big)\,\left(r_{2}-r_{3}\right)}}\)}+{6r_2r_3\big(r_1^3-r_2^3\big)}}{\mp 4{\tau\big(r_1^3-r_2^3\big)}\sqrt{{\Xi_4}^2-{4r_1r_2r_3\big(r_1^3-r_2^3\big)\,\big(r_3^3-r_4^3\big)\,\left(r_{2}-r_{3}\right)}}}{v_{r}}r_1^2.
\end{aligned}
\end{equation}
The proof is complete.
\end{proof}

\begin{rem} 
	(i) In the proof of Theorems \ref{single_reson_frequency}--\ref{dual_reson_frequency}, we can see that the $-\overline{\omega_{N,i}}$, $i =1,\ldots,N_\rmr$, for $N=1,\ldots,4$, is also a subwavelength resonant frequency, which is consistent  with Lemma \ref{symmetric_res}.
	The reason for choosing the real part to be positive is to give a physical meaning to a complex subwavelength resonant frequency.
	The positive real part represents the frequency of oscillation and the negative  imaginary part  corresponds to the fact that energy is lost to the far field with the magnitude describing the rate of attenuation.
	
	(ii) It is known that as the CVR of the shell increases, the  a larger blueshift occurs, i.e., if $\frac{r_1}{r_1^3-r_2^3}\leq \frac{r_3}{r_3^3-r_4^3}$ there holds $\Re \omega_{2,\rm{OS}}\leq \Re \omega_{2,\rm{IS}}$, where $\omega_{2,\rm{OS}}$ and $\omega_{2,\rm{IS}}$ are the subwavelength resonant frequencies of outer and inner shell, respectively. The subwavelength resonant response of the dual-resonator models can be understood as an interaction and hybridization of the response of the two individual shells, which results in mode splitting. It can be seen that $\Re \omega_{4,1}<\Re \omega_{2,\rm{OS}}\leq \Re \omega_{2,\rm{IS}}<\Re \omega_{4,2}$.

\end{rem}

\section{Numerical computations}\label{sec5}

In this section, we conduct numerical simulations to validate our theoretical results from the previous sections. We begin by analyzing mode splitting in multi-layer concentric spheres. Furthermore, it is essential to examine the eigenmodes $u_j$
associated with each eigenfrequency $\omega_{N,j}$.



\subsection{Mode splitting}\label{subsecMS}
In this subsection, we compute the eigenfrequencies. To validate the eigenfrequency formulas in Theorems \ref{single_reson_frequency}--\ref{dual_reson_frequency}, we first compute the characteristic value $\omega_N^{(c)}$ of $\bm{A}_N(\omega, \delta)$ in \eqref{MLCBM}. We compare $\omega_N^{(c)}$ with the exact eigenfrequency formulas, denoted by $\omega_N^{(e)}$, over a range of appropriate values of $\delta$ to assess the accuracy of the formulas.

To conduct the analysis within the appropriate regime, as described in Section \ref{sec3}, we set $\rho_\rmr = \kappa_\rmr = 1$ and $\rho = \kappa = 1/\delta$, where $\delta \in \{1/100,1/1000,1/6000,1/10000\}$.
Setting $ f(\omega) = \det(\bm{A}_N(\omega, \delta))$, we have that calculating $\omega_N^{(c)}$ is equivalent to determining the following complex root-finding problem
\begin{equation}\label{CRFP}
\omega_N^{(c)} = \min\limits_{\omega \in \mathbb{C}}\{\omega| \ f(\omega) = 0\},
\end{equation}
which can be calculated by using Muller's method \cite{AK_book2018}.
We consider the case that radius of layers are equidistance.  For $N$-layer structure, set
\begin{equation}\label{str01}
r_i=(N-i+1), \quad i=1, 2, \ldots N.
\end{equation}
The eigenfrequencies $\omega_4^{(c)}$ and $\omega_4^{(e)}$ for four-layer structures, along with the corresponding  total relative errors for specific values of $\delta$, are presented  in Table \ref{table-results}.
It is evident that the total relative error decreases as $\delta$ becomes smaller, thereby confirming the accuracy of the formulas   \eqref{4Lw1}--\eqref{4Lw2}.
In particular, we observe  that for $\delta = 1/6000$,  the difference between $\omega^{(c)}$ and $\omega^{(e)}$ is negligible with a total relative error of only $0.0229\%$.

By applying Muller's method \cite{AK_book2018} to solve \eqref{CRFP}, we can find, for each fixed $\delta > 0$, that the $N$-layer structure composed of $N_\rmr$ subwavelength nested resonators possesses  $N_\rmr$ eigenfrequencies with positive real parts and negative imaginary part. The eigenfrequencies of 50-layer subwavelength nested resonators are shown in Figure \ref{50layerd1_6000}. It is observed that these eigenfrequencies lie in the lower half of the complex plane and exhibit symmetry about the imaginary axis. Notably, the imaginary part of the lowest frequency is greater than that of the other frequencies. This may have practical implications for applications such as low-pass filters in acoustic wave processing.

\afterpage{\setlength{\tabcolsep}{12pt}
\begin{table}
	\centering
	\begin{tabular}{cccc}
		\toprule
		$\delta$ & $\omega_4^{(c)}$ & $\omega_4^{(e)}$ & \text{Total relative error} \\
		\midrule
		$1/100$ & 0.0513551 -0.0052161i   & 0.0517470 -0.0052491i    & 1.3691\% \\
		& 0.1754137 -0.0012548i   & 0.1764784 -0.0012374i    &  \\
		$1/1000$ & 0.0163513 -0.0005246i   & 0.0163638 -0.0005249i   & 0.1371\% \\
		& 0.0557735 -0.0001239i & 0.0558074 -0.0001237i &  \\
		$1/6000$ & 0.0066797 -0.0000875i    & 0.0066805 -0.0000875i    & 0.0229\% \\
		& 0.0227810 -0.0000206i  & 0.0227833 -0.0000206i  &  \\
		$1/10000$ & 0.0051743 -0.0000525i &0.0051747 -0.0000525i& 0.0137\% \\
		& 0.0176468 -0.0000124i&0.0176478 -0.0000124i&  \\
		\bottomrule
	\end{tabular}
	\caption{A comparison between the characteristic value $\omega_4^{(c)}$ of $\mathcal{A}(\omega, \delta)$ and the the eigenfrequencies formulas  \eqref{4Lw1}--\eqref{4Lw2}, over several values of $\delta$.}
	\label{table-results}
\end{table}
\begin{figure}[h]
	\centering
	\includegraphics[scale=0.5]{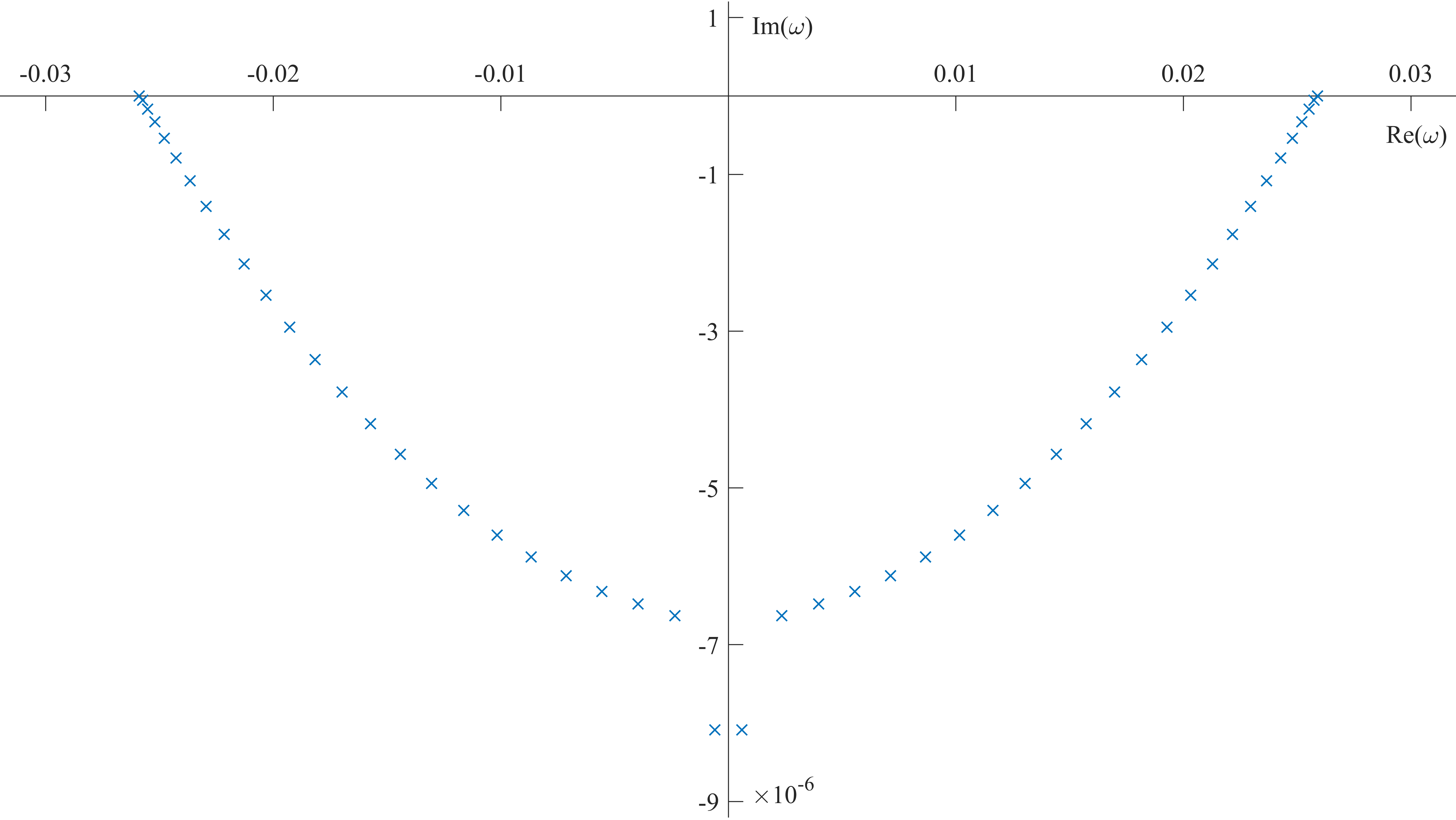}
	\caption{ The subwavelength resonant frequencies, plotted in the complex plane, of the $50$-layer structure composed of $25$ subwavelength nested resonators designed by \eqref{str01} with $\delta = 1/6000$.}\label{50layerd1_6000}
\end{figure}
}

\subsection{Resonant modes}

In subsection \ref{subsecMS} above, we see that for the $N$-layer concentric ball consisting of $N_\rmr$ subwavelength nested resonators, there exist $N_\rmr$ eigenfrequencies with positive real parts and negative imaginary parts. Next, we aim to analyze  the distribution of the eigenmodes $u_j$ associated with each eigenfrequency  $\omega_{N,j}$ (abbreviated as $\omega_j$).

For easier visualization of results, we focus on seven- or eight-layer structures consisting of quadruple-resonators.
We also consider that the radii of the layers decrease at the same scale $s$, that is,
\begin{equation}\label{str02}
r_{i+1}=sr_i, \quad i=1, 2, \ldots N-1.
\end{equation}
Let $r_1=N$ and $s=0.8$.
The four eigenfrequencies for the eight-layer concentric ball, designed using \eqref{str01}, and the seven-layer concentric ball, designed using \eqref{str02}, are shown in Figures \ref{8layer_d1_6000} and \ref{7layer_s08_6000}, respectively.
The eigenmodes have been normalized such that
$\int_{D_{\mathrm{r}}} |u_j|^2\mathrm{d}x = 1$
for each $j = 1,2,3, 4.$

The eigenfrequencies $\omega_j$ are arranged in ascending order of their real parts. The corresponding eigenmodes inherit the symmetry of the nested resonators and exhibit increasingly oscillating pattern.
Intriguingly, the phenomenon of field concentration is evident in the lower plot of Figure \ref{7layer_s08_6000}, where the gradient of the solution may become arbitrarily large as the distance between the two resonators approaches zero, which are the counterparts of the well-known
gradient blowup \cite{DKLZ_AAMM2024,ADY_MMS2020,LZ_MMS2023,DFLMMS22} in two nearly touching separated resonators. This phenomenon is a central topic in the theory of composite materials. It is also notable that the solution
remains approximately constant within each resonator.
This behavior arises because, at leading order, the solution takes the form \eqref{Helm-solution}, represented as $\mathcal{S}_{\Gamma_j} [\varphi_j]+\mathcal{S}_{\Gamma_{j+1}} [\psi_{j+1}]$ for odd $j$. According to Lemma \ref{A00properties}, this expression is constant for $\varphi_j\in \ker (-\frac{1}{2}I+ \mathcal{K}_{\Gamma_j}^{*})$ and $\psi_{j+1}=0$.

\afterpage{
\begin{figure}[htbp]
	\centering
	\includegraphics[scale=0.73]{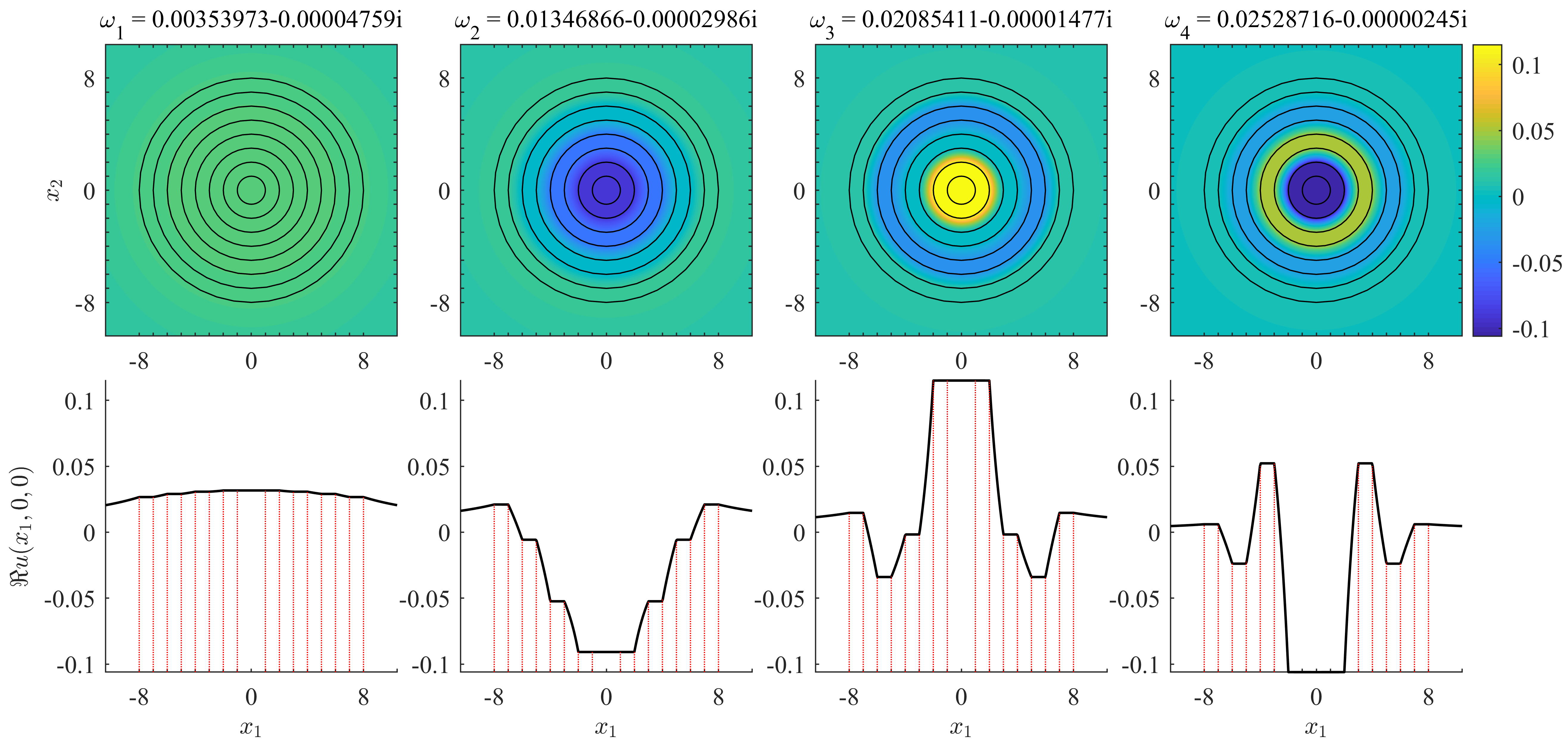}
	\caption{The acoustic pressure eigenmodes $u_1,u_2,u_3,u_4$ for the eight-layer concentric ball designed by \eqref{str01}. Each pair of plots corresponds to one of the four eigenfrequencies. The upper plot displays a contour plot of the function $\Re u_k(x_1, x_2,0)$, with the eight-layer concentric ball designed by \eqref{str01} represented as solid black lines. The lower plot shows the cross section of the upper plot, taken along the line $x_2 = 0$ (passing through the centres of the multi-layer structures). Additionally, red dotted lines represent vertical lines at the coordinates of the radius.}\label{8layer_d1_6000}
\end{figure}
\begin{figure}[htbp]
	\centering
	\includegraphics[scale=0.73]{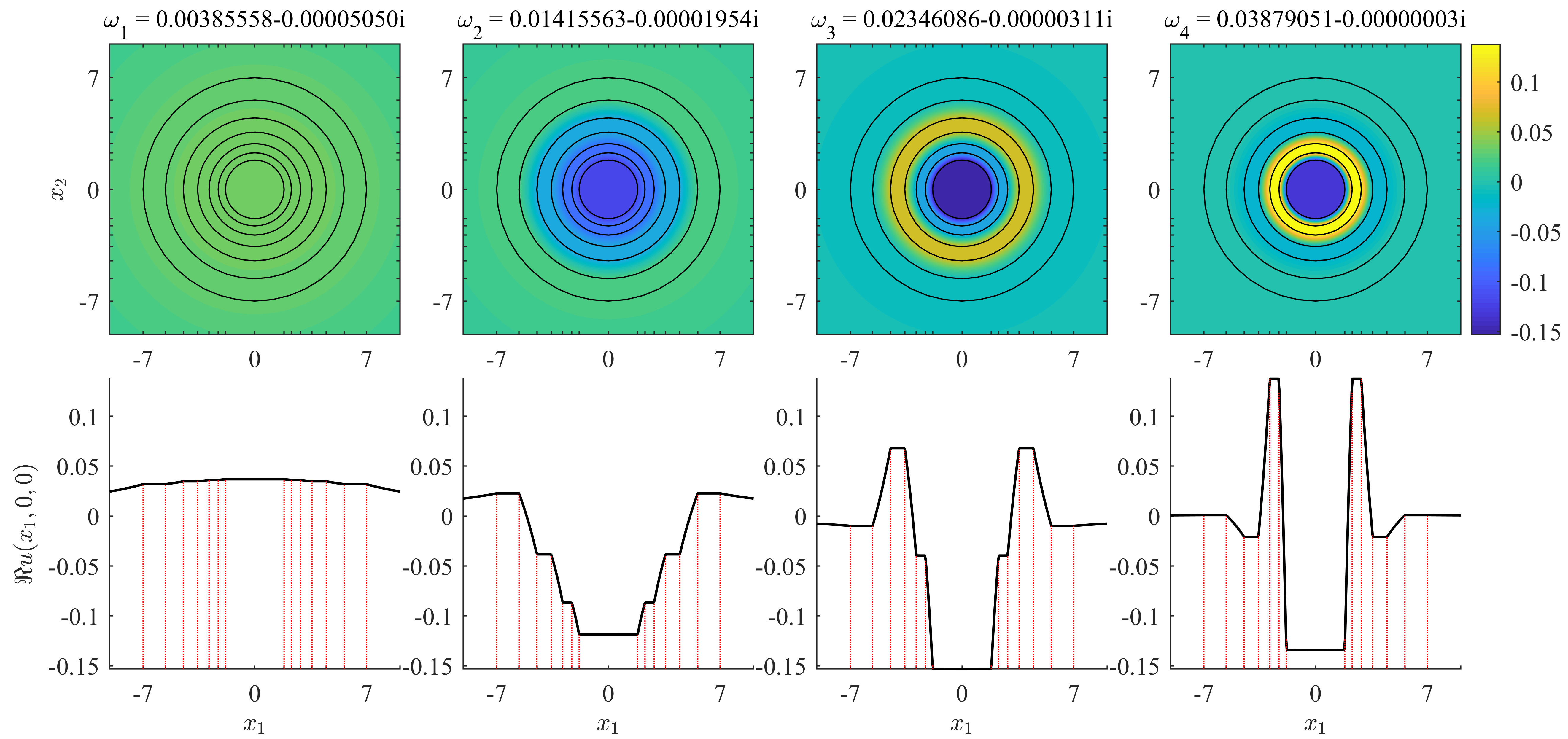}
	\caption{The acoustic pressure eigenmodes $u_1,u_2,u_3,u_4$ for the seven-layer concentric ball designed by \eqref{str02}. Each pair of plots corresponds to one of the four eigenfrequencies. The upper plot displays a contour plot of the function $\Re u_k(x_1, x_2,0)$, with the seven-layer concentric ball designed by \eqref{str02} represented as solid black lines. The lower plot shows the cross section of the upper plot, taken along the line $x_2 = 0$ (passing through the centres of the multi-layer structures). Additionally, red dotted lines represent vertical lines at the coordinates of the radius.}\label{7layer_s08_6000}
\end{figure}}

\section{Concluding remarks}\label{sec6}
In this paper, we developed a general mathematical framework to study resonant phenomenon in acoustics within multi-layer high-contrast structures. Using layer potential techniques and Gohberg-Sigal theory, we demonstrate a powerful method for analyzing metamaterials with nested geometry, revealing that the number of resonant modes increases with the number of resonators. This indicates a rich landscape of acoustic behavior that can be tailored through careful design, with significant implications for applications in manipulating  waves propagation at subwavelength scales. We mention that the idea can be extended to the two-dimensional case with some technical adjustments. Moreover,  We shall derive some similar results regarding the subwavelength resonance in multi-layer high contrast solid structures in forthcoming works. It is worth emphasizing that the approach developed in this paper can be applied to study the mode splitting of other subwavelength resonators such as multi-layer plasmonic structures.

\section*{Acknowledgment}
The work of Y. Deng was supported by NSFC-RGC Joint Research Grant No. 12161160314.
The work of H. Li was substantially supported by NSFC grant 12401561 and Research Start Fund 5333100432.
 The work of H. Liu was supported by the NSFC/RGC Joint Research Scheme, N\_CityU101/21; ANR/RGC Joint Research Scheme, A\_CityU203/19; and the Hong Kong RGC General Research Funds (projects 11311122, 11304224 and 11300821).


\begin{thebibliography}{99}
	


\bibitem{ACLJFA23}
{\sc H. Ammari, Y. Chow, and H. Liu},
{\em Quantum ergodicity and localization of plasmon resonances}, J. Funct. Anal., 285 (2023), 109976.


\bibitem{Ammari2013}
{\sc H.~Ammari, G.~Ciraolo, H.~Kang, H.~Lee, and G.~W. Milton}, {\em Spectral theory of a Neumann-Poincar{\'{e}}-type operator and analysis of cloaking due to anomalous localized resonance}, Arch. Ration. Mech. Anal., 208 (2013), 667--692.

\bibitem{AD_book2024}
{\sc  H. Ammari, and B. Davies}, {\em Metamaterial Analysis and Design: A Mathematical Treatment of Cochlea-inspired Sensors}, (De Gruyter, Berlin, 2024).

\bibitem{ADY_MMS2020}
{\sc H. Ammari, B. Davies, and S. Yu},
{\em Close-to-touching acoustic subwavelength resonators: eigenfrequency separation and gradient blow-up},
Multiscale Model. Simul., 18 (2020), 1299--1317.

\bibitem{Ammari2016}
{\sc H.~Ammari, Y.~Deng, and P.~Millien}, {\em Surface plasmon resonance of nanoparticles and applications in	imaging},  Arch. Ration. Mech. Anal., 220 (2016), 109--153.

\bibitem{AFGLZ_AIHPCAN}
{\sc H. Ammari, B. Fitzpatrick, D. Gontier, H. Lee, and H. Zhang},
{\em Minnaert resonances for acoustic waves in bubbly media},
 Ann. Inst. H. Poincar\'e C Anal. Non Lin\'eaire, 35 (2018), 1975--1998.

\bibitem{AFHLY_JDE2019}
{\sc H. Ammari, B. Fitzpatrick, E.O. Hiltunen, H. Lee, and S. Yu}, {\em Subwavelength resonances of encapsulated bubbles},  J. Differential Equations, 267 (2019), 4719--4744.

\bibitem{AFHLY_SIMA2020}
{\sc H. Ammari, B. Fitzpatrick, E.O. Hiltunen, H. Lee, and S. Yu},
{\em Honeycomb-lattice Minnaert bubbles}, SIAM J. Math. Anal., 52 (2020), 5441--5466.

\bibitem{AK_book2018}
{\sc H. Ammari, B. Fitzpatrick, H. Kang, M. Ruiz, S. Yu, and H. Zhang},  {\em Mathematical and Computational Methods in Photonics and Phononics},  (American Mathematical Society, Providence, RI, 2018).

\bibitem{AFLYZ_JDE2017}
{\sc H. Ammari, B. Fitzpatrick, H. Lee, S. Yu, and H. Zhang},
{\em Subwavelength phononic bandgap opening in bubbly media}, J. Differential Equations, 263 (2017), 5610--5629.




\bibitem{AFLYZQAM19}
{\sc H.~Ammari, B.~Fitzpatrick, H.~Lee, S.~Yu, and H.~Zhang},  {\em Double-negative acoustic metamaterials}, Quart. Appl. Math., 77 (2019), 767--791.

\bibitem{Ammari2007}
{\sc H. Ammari and H. Kang},
{\em Polarization and Moment Tensors with Applications to Inverse Problems and Effective Medium Theory}, (Springer-Verlag, New York, 2007).

\bibitem{ALLZ526} {H. Ammari, B. Li, H. Li and J. Zou}, {\em Fano resonances in all-dielectric electromagnetic metasurfaces},  {Multiscale Model. Simul.}, { 22} (2024), 476-526.


\bibitem{AmmariSIMA17}
{\sc H.~Ammari and H.~Zhang}.
{\em Effective medium theory for acoustic waves in bubbly fluids near minnaert resonant frequency}, SIAM J. Math. Anal., 49 (2017), 3252--3276.

\bibitem{BLLW_ESAIM2020}
{\sc E. Bl\aa sten, H. Li, H. Liu, and Y. Wang}, {\em Localization and geometrization in plasmon resonances and geometric structures of Neumann-Poincar\'e eigenfunctions}, ESAIM Math. Model. Numer. Anal., 54 (2020), 957--976.

\bibitem{BZ_RMI2019}
{\sc E. Bonnetier and H. Zhang},
{\em Characterization of the essential spectrum of the Neumann-Poincar\'e operator in 2D domains with corner via Weyl sequences},
Rev. Mat. Iberoam., 35 (2019),  925--948.



\bibitem{CMJW_SV2018}
{\sc M. Chen, D. Meng, H. Jiang,and Y. Wang},
{\em Investigation on the band gap and negative properties of
concentric ring acoustic metamaterial}, Shock Vib., 12 (2018), 1369858.

%

\bibitem{CK_book}
{\sc D. Colton and R. Kress}, {\em Inverse Acoustic and Electromagnetic Scattering Theory}, (Springer, Cham, 2013).

\bibitem{CAJASA1989}
{\sc K. Commander and A. Prosperetti}, {\em Linear pressure waves in bubbly liquids: Comparison between theory and experiments},  J. Acoust. Soc. Amer., 85 (1989), 732--746.




\bibitem{DFLMMS22}
{\sc Y.~Deng, X.~Fang, and H.~Liu}, {\em Gradient estimates for electric fields with multiscale inclusions in the quasi-static regime}, Multiscale Model. Simul., 20 (2022), 641--656.


\bibitem{DKLZ_AAMM2024}
{\sc Y. Deng,  L. Kong, H. Liu, and L. Zhu}, {\em On field concentration between nearly-touching multiscale inclusions in the quasi-static regime}, Adv. Appl. Math. Mech., 16 (2024), 1252--1276.

\bibitem{DLL242} {Y. Deng, H. Li and H. Liu}, {\em Spectral properties of Neumann-Poincar\'e operator and anomalous localized resonance in elasticity beyond quasi-static limit}, \textsl{Journal of Elasticity}, {140}(2020), 213--242.

\bibitem{DLbook2024}
{\sc Y. Deng and H. Liu}, \emph{Spectral Theory of Localized Resonances and Applications}, (Springer, Singapore, 2024).

\bibitem{DLZJMPA21}
{\sc Y. Deng, H. Liu, and G.-H. Zheng},
 {\em Mathematical analysis of plasmon resonances for curved nanorods}, J. Math. Pure Appl., 153 (2021), 248--280.



\bibitem{DSWYZ_APL2021}
{\sc H. Duan, X. Shen, E. Wang, F. Yang, X. Zhang, and Q. Yin}, {\em Acoustic multi-layer Helmholtz resonance metamaterials with multiple adjustable absorption peaks}, Appl. Phys. Lett., 118 (2021), 241904.


\bibitem{DB_IEEE2011}
{\sc A. Doinikov and A. Bouakaz}, {\em Review of shell models for contrast agent microbubbles},  IEEE Trans. Ultrason. Ferroelectr. Freq. Control, 58 (2011), 981--993.


\bibitem{dyatlov2019mathematical}
S. Dyatlov and M. Zworski, {\em Mathematical Theory of Scattering Resonances}, Grad. Stud. Math. 200,
American Mathematical Society, Providence, RI, 2019.

\bibitem{FangdengMMA23}
{\sc X. Fang and Y. Deng}, {\em On plasmon modes in multi-layer structures}, Math. Methods Appl. Sci., 46 (2023), 18075--18095.

\bibitem{FDLMMA15}
{\sc X.~Fang, Y.~Deng, and J.~Li},
{\em Plasmon resonance and heat generation in nanostructures}, Math. Methods Appl. Sci., 38 (2015), 4663--4672.



\bibitem{JKMA23}
{\sc Y.-G. Ji and H. Kang},
{\em Spectral properties of the Neumann-Poincar\'e operator on
rotationally symmetric domains},  Math. Ann., 387 (2023), 1105--1123.


\bibitem{KKLSY_JLMS2016}
{\sc H. Kang, K. Kim,  H. Lee, J. Shin, and S. Yu},
{\em Spectral properties of the Neumann-Poincar\'e operator and uniformity of estimates for the conductivity equation with complex coefficients}, J. Lond. Math. Soc.,  93 (2016),  519--545.



\bibitem{KZDF}
{\sc L. Kong, L. Zhu, Y. Deng, and X. Fang}, {\em Enlargement of the localized resonant band gap by using multi-layer structures},  J. Comput. Phys., 518 (2024), 113308.

\bibitem{KMKG_JVA2017}
{\sc A. Krushynska, M. Miniaci, V. Kouznetsova, and M. Geers},
{\em Multilayered inclusions in locally resonant metamaterials: Two-dimensional versus three-dimensional modeling},
J. Vib. Acoust. 139 (2017), 024501.


\bibitem{LWSZ_NM2011}
{\sc Y. Lai, Y. Wu, P. Sheng, and Z.  Zhang}, {\em Hybrid elastic solids}, Nature Materials, 10 (2011), 620--624.

\bibitem{LPDV_PRE2007}
{\sc H. Larabi, Y. Pennec, B. Djafari-Rouhani, and J. Vasseur},
{\em Multicoaxial cylindrical inclusions in locally resonant phononic crystals}, Phys. Rev. E, 75, (2007) 066601.

\bibitem{LVAPL09}
{\sc V.~Leroy, A.~Bretagne, M.~Fink, H.~Willaime, P.~Tabeling and A.~Tourin}, {\em Design and characterization of bubble phononic crystals}, Appl. Phys. Lett.,  95 (2009), 171904.

\bibitem{LSPS_JASM2008}
{\sc V. Leroy, A. Strybulevych, J. H. Page, and M. G. Scanlon},
{\em Sound velocity and attenuation in bubbly gels measured by
transmission experiments},
J. Acoust. Soc. Amer., 123, (2008), 1931--1940.

\bibitem{L1272} {\sc H. Li}, {\em Recent progress on the mathematical study of anomalous localized resonance in elasticity}, {Electron. Res. Arch.},  {28} (2020), 1257--1272.

\bibitem{LLPRSA18}
{\sc H.~Li and H.~Liu}, {\em On anomalous localized resonance and plasmonic cloaking beyond the quasi-static limit}, Proc. R. Soc. A, 474 (2018), 20180165.

\bibitem{LLZSIAM2022}
{\sc H.~Li,  H.~Liu, and J. Zou},
{\em Minnaert resonances for bubbles in soft elastic materials},
SIAM J. Appl. Math., 82 (2022), 119--141.

\bibitem{LL7527} {H. Li and L. Xu}, {\em Resonant modes of two hard inclusions within a soft elastic material and their stress estimate}, arXiv:2407.19769.

\bibitem{LZ_MMS2023}
{\sc H. Li and Y. Zhao}, {\em The interaction between two close-to-touching convex acoustic subwavelength resonators}, Multiscale Model. Simul., 21 (2023), 804--826.

\bibitem{LZArxiv}
{\sc H. Li and J. Zou},
{\em Mathematical theory on dipolar resonances of hard inclusions within a soft elastic material}, arXiv:2310.12861.

\bibitem{LZScience}
{\sc Z. Liu, X. Zhang, Y. Mao, Y.  Zhu, Z. Yang,  T. Chan, and P. Sheng}, {\em Locally resonant sonic
materials}, Science, 289 (2000), 1734--1736.










\bibitem{MBKPD_PNAS2016}
{\sc K.H. Matlack, A. Bauhofer, S. Kr\"odel, A. Palermo, C. Daraio}, {\em Composite 3D-printed metastructures for low-frequency and broadband vibration absorption}, Proc. Natl. Acad. Sci. USA, 113 (2016), 8386--8390.


\bibitem{Min_1933}
{\sc M. Minnaert}, {\em On musical air-bubbles and the sounds of running water},  Philos. Mag., 16 (1933), 235--248.

\bibitem{CPBE2018}
{\sc T. Ngo, A. Kashani, G. Imbalzano, K. Nguyen, and D. Hui},
{\em Additive manufacturing (3D printing): A review of materials, methods, applications and challenges}, Compos. Part B: Eng., 143 (2018), 172--196.

\bibitem{PRHN2003SCI}
{\sc E.~Prodan, C.~Radloff, N.~Halas, and P.~Nordlander},
{\em A hybridization model for the plasmon response of complex nanostructures},  Science, 302(2003), 419--422.


\bibitem{PVDD_SCR2010}
{\sc Y. Pennec, J. Vasseur, B. Djafari-Rouhani, L. Dobrzy\'nski, and P. Deymier}, {\em Two-dimensional phononic crystals: Examples and applications}, Surf. Sci. Rep., 65 (2010), 229--291.






\bibitem{SPMW_AA2023}
{\sc G. Szczepa\'nski, M.  Podle\'sna, L. Morzynski and A. W{\l}udarczyk}, {\em Investigation of the acoustic properties of a metamaterial with a multi-ging structure}, Arch. Acoust., 48 (2023), 497--507.


\bibitem{JPT2001}
{\sc J.-P. Tignol}.
 {\em {Galois' Theory of Algebraic Equations}},
(World Scientific Publishing, River Edge, NJ, 2001).

\bibitem{ZH_PRB2009}
{\sc X. Zhou and G. Hu}, {\em Analytic model of elastic metamaterials with local resonances}, Phys. Rev. B, 79 (2009), 195109.

\bibitem{YA_SIAMREV18}
{\sc S.~Yu and H.~Ammari}, {\em Plasmonic interaction between nanospheres}, SIAM Rev., 60 (2018), 356--385.

\bibitem{YA_PNAS19}
{\sc S.~Yu and H.~Ammari}, {\em Hybridization of singular plasmons via transformation optics}, Proc. Natl. Acad. Sci. USA, 116 (2019), 13785--13790.

\end {thebibliography}
\end{document}